\documentclass[journal]{IEEEtran}

\usepackage{amsmath,amssymb,amsfonts,amsthm,bm}
\usepackage{graphicx,epsfig,color}
\usepackage{cite}
\usepackage{hyperref}
\usepackage{url}
\usepackage{tikz}
\usepackage{pgfplots}
\usepgfplotslibrary{groupplots}
\pgfplotsset{compat=newest}
\usepackage{mathtools}
\mathtoolsset{showonlyrefs}

\pgfplotsset{
  every axis/.append style={
    tick label style={font=\footnotesize},
    label style={font=\footnotesize},
    legend style={font=\footnotesize}, 
    title style={font=\footnotesize}
  },
  every axis plot/.append style={
    mark options={scale=0.5} 
  }
}
\tikzset{semithick/.style={line width=1pt}}

\renewcommand{\bm}{\pmb}
\newcommand{\T}{{\top}}
\newcommand{\nr}{{n_\mathrm{r}}}
\newcommand{\nt}{{n_\mathrm{t}}}
\newcommand{\xv}{\bm{x}}
\newcommand{\yv}{\bm{y}}
\newcommand{\zv}{\bm{z}}
\newcommand{\uv}{\bm{u}}
\newcommand{\hv}{\bm{h}}
\newcommand{\Hm}{\bm{H}}
\newcommand{\Id}{\bm{I}}

\newcommand{\Vol}{\mathrm{Vol}}
\newcommand{\Ball}{{\mathcal{B}}}
\newcommand{\Sphere}{{\mathcal{S}}}

\newcommand{\Jm}{\bm{J}}
\newcommand{\Mm}{\bm{M}}
\newcommand{\E}{\mathbb{E}}
\newcommand{\ind}{\mathbf{1}}
\newcommand{\Q}{\mathsf{Q}}

\newcommand{\Xm}{\bm{X}}
\newcommand{\Ym}{\bm{Y}}

\newcommand{\thetav}{{\bm{\theta}}}

\newcommand{\Xc}{\mathcal{X}}
\newcommand{\Sc}{\mathcal{S}}
\newcommand{\Yc}{\mathcal{Y}}
\newcommand{\Bc}{\mathcal{B}}
\newcommand{\Vc}{\mathcal{V}}

\newcommand{\JF}{\mathsf{JF}}
\newcommand{\SM}{\mathsf{M}}
\newcommand{\GR}{\bm{\Gamma}}
\newcommand{\clip}{\mathsf{clip}}

\let\oldmid\mid
\renewcommand{\mid}{\!\oldmid\!}
\newcommand{\comment}[1]{ {\color{blue} #1} }

\newcommand{\revision}[1]{{\color{blue}#1}}

\renewcommand{\revision}[1]{#1}

\newcommand{\newrevision}[1]{{\color{blue}#1}}
\renewcommand{\newrevision}[1]{#1}

\allowdisplaybreaks

\newtheorem{proposition}{Proposition}
\newtheorem{lemma}{Lemma}
\newtheorem{theorem}{Theorem}
\newtheorem{corollary}{Corollary}
\newtheorem{remark}{Remark}

\begin{document}

\title{From Bayesian Asymptotics \\to General Large-Scale MIMO Capacity}
\author{Sheng~Yang,~\IEEEmembership{Member,~IEEE,} and Richard
Combes,~\IEEEmembership{Member,~IEEE}
\thanks{S. Yang and R. Combes are with the laboratory of signals and systems at
         CentraleSup\'elec-CNRS-Universit\'e Paris-Sud, 91192,
         Gif-sur-Yvette, France. Email:\{sheng.yang, richard.combes\}@centralesupelec.fr} 
\thanks{This work was supported in part by the Agence Nationale de la Recherche (ANR)
through the France 2030 project ANR-PEPR Networks of the
Future under grant agreement NF-YACARI 22-PEFT-0005.}
}

\maketitle

\renewcommand{\comment}[1]{#1}

\begin{abstract}
We present a unifying framework that bridges Bayesian asymptotics and
information theory to analyze the asymptotic Shannon capacity of general large-scale
MIMO channels including ones with nonlinearities or imperfect hardware. 
\comment{We derive both an analytic capacity formula and an asymptotically optimal input distribution in the large-antenna regime, each of which depends solely on the single-output channel's Fisher information through a term we call the \emph{(tilted) Jeffreys factor}.} We demonstrate how our method applies broadly to scenarios with clipping,
coarse quantization (including 1-bit ADCs), phase noise, fading with imperfect
CSI, and even optical Poisson channels. Our asymptotic
analysis motivates a practical approach to constellation design via a
compander-like transformation. \comment{Furthermore, we introduce a low-complexity receiver structure
that approximates the log-likelihood by quantizing the channel outputs into
finitely many bins, enabling near-capacity performance with computational complexity independent
of the output dimension.} Numerical results confirm that the proposed
method unifies and simplifies many previously intractable MIMO capacity
problems and reveals how the Fisher information alone governs the channel’s
asymptotic behavior.
\end{abstract}

\section{Introduction}
\label{sec:intro}
Ever-increasing demand for high data rates has made large-scale (or
{massive}) MIMO systems a leading candidate in modern wireless
architectures. By deploying a large number of antennas,
these systems can substantially improve spatial diversity and spectral
efficiency. However, analyzing the Shannon capacity of such large-scale systems is
complex, especially when hardware impairments (e.g., clipping or
saturation effects), low-precision analog-to-digital converters (ADCs),
or other nonlinearities are present. In addition, certain channel models (e.g.,
optical Poisson channels) are inherently nonlinear, further challenging
conventional approaches primarily suited to idealized linear MIMO
channels.

When the number of receive antennas grows, a powerful link emerges between
{Bayesian asymptotics} and {information theory}. This link is closely
tied to the concept of \emph{Bayesian redundancy}, wherein the
minimum expected redundancy --- defined as the Kullback-Leibler (KL) divergence
from the true distribution within a parametric family to the estimated one ---
coincides with the mutual information between the parameter and the observations.
The \emph{worst-case} or \emph{maximin} Bayesian redundancy thus equals the maximum
mutual information over all prior distributions of the parameter. In their
seminal paper \cite{CB94}, Clarke and Barron showed that, in the large-sample regime,
this maximum is achieved by \emph{Jeffreys prior}, which is proportional
to the square root of the Fisher information determinant. In a
communication setting --- where the channel input can be viewed as the parameter of
the output distribution --- Clarke--Barron’s result provides a natural
route to derive the asymptotic channel capacity, as demonstrated in
\cite{YC2024} for 1-bit MIMO channels in the large-antenna regime.

In this paper, we present a unified, step-by-step method that leverages Bayesian
asymptotics to analyze a broad class of large-scale MIMO channels. In essence,
when each receive antenna provides an independent observation drawn from a
distribution parameterized by the input, the channel’s capacity can be
characterized up to a vanishing term as the number of antennas grows large. We
show that the asymptotically optimal input distribution is {Jeffreys
prior}, determined solely by the Fisher information matrix of the single-letter
(per-antenna) distribution. Since this Fisher information is 
straightforward to compute \revision{analytically or numerically in most cases}, 
the framework readily applies to many scenarios of
interest. \comment{Our contributions are fourfold:}
\begin{itemize}
  \item First, we derive a systematic recipe for computing the asymptotic
  capacity of large-scale MIMO channels under both {peak-} and {average-power} constraints. Specifically, we extend Clarke and Barron’s result
  \cite{CB94} by incorporating an additional second-moment constraint typical
  in communication settings. A key outcome is a compact capacity expression in
  which a single \emph{Jeffreys factor} encapsulates both the power constraints
  and the channel distribution, all through the Fisher information of the
  per-antenna output.
  \item Second, we illustrate how this recipe addresses capacity problems that
  are intractable under conventional methods. Starting with single-input
  multiple-output (SIMO) channels that involve clipping, quantization, or phase
  noise, we then show how fading (including imperfect channel state
  information) and optical Poisson models also fit naturally into the proposed
  framework. Notably, we demonstrate how introducing dithering (i.e., artificially
  shifting signals prior to quantization) can boost the capacity of a multi-antenna
  channel with 1-bit ADCs.
  \item Third, inspired by the asymptotically optimal Jeffreys prior, we
  propose a systematic and universal compander-like transformation for practical constellation
  design. By applying a smooth nonlinear mapping to a uniform grid, one obtains
  a constellation that can significantly outperform conventional choices. We
  demonstrate this numerically in a multi-antenna system with low-resolution
  ADCs.
\item \comment{
Finally, we introduce a low-complexity receiver architecture that approximates the log-likelihood by quantizing the channel outputs into finitely many bins. We prove --- by analyzing the scaling of its capacity loss --- that this approach achieves near-capacity performance while maintaining computational complexity independent of the output dimension, thereby further reinforcing the practicality of our Bayesian-asymptotic framework.
}
\end{itemize}

Standard MIMO systems --- i.e., linear channels with infinite-precision receivers and
only average-power constraints --- have been well investigated since the pioneering
works \cite{Telatar99,Foschini}. For channels with nonlinearities such as
clipping or saturation~\cite{ochiai2002performance, sabbaghian2013analysis}, quantized output
\cite{Singh09,koch2013low,mezghani2020low}, or phase noise
\cite{Lapidoth02,Durisi14,GhozlanKramer15,YangShamai17}, and for channels
with only partial channel state information~(CSI) \cite{marzetta1999capacity,ZhengTse2002Grassman,Moser}, capacity results
are generally known only in special asymptotic-SNR regimes. In the optical
domain, Poisson channel models \cite{Shamai90,LapidothMoserPoisson09} yield
capacity bounds but no simple closed-form expressions. 

Motivated by large-scale or massive MIMO systems~\cite{Marzetta10}, our
work focuses on the regime where the number of receive antennas grows large. We
extend Clarke--Barron's Bayesian asymptotic result~\cite{CB94} to incorporate
both peak- and average-power constraints into a wide class of MIMO models,
yielding explicit, and surprisingly simple, capacity expressions valid for
large arrays. Consequently, even for \emph{standard} linear MIMO
channels, our approach provides a new capacity characterization that includes
both  peak- and average-power constraints in a tractable form. Beyond this,
our framework naturally accommodates nonlinear transceivers and channels with
imperfect CSI, as well as other nonlinear channel models. Just as the SNR serves
as the canonical surrogate for capacity in the linear regime, the
\emph{Jeffreys factor} now emerges as the key surrogate for capacity in these
more general large-scale MIMO settings. 

From a statistical perspective, the notion of a \emph{Jeffreys prior} and its
variants~\cite{kass1996selection} is well known 
as a noninformative prior that is invariant by any change of coordinates in the parameter space.
 In \cite{CB94}, Clarke and Barron formally proved that it
is asymptotically minimax and maximin in the large-sample regime. Remarkably,
the same distribution that is least favorable in a statistical sense turns out
to be the \emph{best} input distribution in a communication sense: it maximizes
the mutual information (i.e., capacity) as the number of receive antennas grows
large. Thus, the minimax principle for KL divergence --- originally
posed in parametric estimation --- coincides with the capacity-achieving design
principle in large-scale MIMO channels.

The rest of the paper is organized as follows. In Section~\ref{sec:problem},
we introduce the channel model and formulate the capacity characterization
problem. Then, we present our main result and give a concise recipe for
obtaining the asymptotic capacity in large-scale MIMO channels.
Sections~\ref{sec:applications} provide a wide range of detailed examples, each accompanied by numerical suggestions to illustrate how the optimal prior~(or Jeffreys factor)
changes with system parameters. Section~\ref{sec:constellation} discusses
practical constellation design inspired by the asymptotically optimal input
distribution. Section~\ref{sec:receiver} analyzes the capacity loss due to approximate log-likelihood computation at the receiver's side. \revision{Section~\ref{sec:noniid} shows how to extend our main results to the general case with correlated outputs.} Section~\ref{sec:conclusion} draws conclusions and suggests open
directions. Technical proofs are deferred to the Appendix to maintain clarity and flow.


\paragraph*{Notation}
We use $\log$ to denote base-2 logarithms and $\ln$ for natural
logarithms. Bold letters (e.g., $\bm{x}$, $\bm{y}$) denote vectors.  The
transpose operator is $(\cdot)^{\T}$, and $\|\cdot\|$ denotes the Euclidean
norm.  \revision{We let $\Ball_{n}$ be the unit $n$-ball, $\Sphere_{n-1}$ be the
unit $(n-1)$-sphere, and $\Vol(\cdot)$ denote volume, $\Ball_n(\bm{x},r)$ an $n$-ball of radius $r$ that is centered at $\bm{x}$.}
The functions \revision{$\phi(x):={1\over\sqrt{2\pi}} e^{-{x^2\over 2}}$ and $Q(x) = \int_x^\infty \phi(t)\, dt$, $x\in \mathbb{R}$,} denote
the pdf and the complementary cdf of the standard Gaussian distribution, respectively. 
\revision{The pdf of a circularly symmetric Gaussian is $\phi_c(x; \sigma^2):={1\over \pi \sigma^2} e^{-{|x|^2 \over \sigma^2}}$, $x\in \mathbb{C}$,} and $\Gamma(\cdot)$ is the Gamma function. We denote $[n] := \{1,\ldots,n\}$. We employ $o(\cdot)$ and $O(\cdot)$ for asymptotic notation \revision{in the usual sense}.  \revision{The set of real, complex, and natural numbers are denoted by $\mathbb{R}$, $\mathbb{C}$, and $\mathbb{N}$, respectively.}

\section{Problem Formulation and Proposed Method}
\label{sec:problem}

\subsection{Channel model}
Consider a memoryless stationary channel \revision{without feedback} with $\nt$ inputs and $\nr$
outputs at each time slot. In $n$ time slots, 
$n=1,2,\ldots$, the channel law factors as 
\begin{equation}
  p(\yv_1,\ldots,\yv_n \mid \xv_1,\ldots,\xv_n) =  
  \prod_{i=1}^n p(\yv_i \mid \xv_i), 
\end{equation}%
where $\xv_i\in\Xc\subset \mathbb{C}^{\nt}$ and $\yv_i\in\Yc^{\nr}$ are the input and output
at the $i$-th time slot, respectively; the conditional probability
density\footnote{\revision{The density function $p(\yv_i | \xv_i)$ is absolutely continous with respect to a dominating measure $\nu(y)$ over $\mathcal{Y}^{\nr}$.} It is a probability mass function~(pmf) when $\nu$ is the counting measure, and a pdf when $\nu$ is the Lebesgue measure.}
$p(\yv_i | \xv_i)$ is time-invariant. 
We assume the input alphabet $\Xc$ is bounded~(e.g., peak-power
constraint) and the input sequence is subject to the average-power\footnote{We assume that the norm is well defined in
$\Xc$.} constraint
\begin{equation}
  {1\over n} \sum_{i=1}^n \|\xv_i\|^2 \le P. 
\end{equation}%
The corresponding Shannon capacity is well defined and admits a single-letter form
\begin{equation}
  C(P) = \max_{p(\xv):\; {\xv\in\Xc\atop\E[\|\xv\|^2] \le P}} I(\xv; \yv), 
  \label{eq:capa}
\end{equation}%
and we denote by $C$ the unconstrained capacity (i.e., when $P$ is
sufficiently large). Moreover, we suppose that 
the receiver's observations at each channel use are~i.i.d.~from the distribution $p(y\mid \xv)$, i.e., 
\begin{equation}
  p(\yv \mid \xv) = \prod_{k=1}^\nr p(y_k \mid \xv),
  \label{eq:iid_output} 
\end{equation}%
where $p(y_k \mid \xv)$ is the same for each receive antenna $k$. This
assumption holds in many scenarios where each antenna experiences an
independent channel or local state (e.g., a random offset), known at the
receiver but not at the transmitter. \revision{A possible relaxation of the i.i.d.~assumption is discussed in Section~\ref{sec:noniid}.}

A particular example is the standard MIMO fading channel
\begin{equation}
  \yv = \Hm  \xv + \zv,
  \label{eq:MIMO}
\end{equation}
where $\zv$ has i.i.d.~noise samples $(z_1,\ldots,z_\nr)$ and
$\Hm$ has i.i.d.~fading vectors $(\hv_1,\ldots,\hv_\nr)$ as its rows.
We assume that $\xv$ is independent of $\Hm$ (i.e., there is no transmitter CSI).
When CSI is unavailable at the receiver, the capacity formula~\eqref{eq:capa}
applies directly; with perfect CSI at the receiver, one can treat $\Hm$ as part
of the output via $\yv \leftarrow (\yv,\Hm)$. Either way, the i.i.d.~structure
in \eqref{eq:iid_output} still holds. Additionally, one may handle block-fading
models by extending $\xv, \yv$ to appropriately sized matrices
$\Xm \in \mathbb{C}^{\nt \times T}$ and $\Ym \in \mathbb{C}^{\nr \times T}$
and letting $\Xc$ be a subset of $\mathbb{C}^{\nt\times T}$ and $\Yc$ be
$\mathbb{C}^{1\times T}$.

\subsection{Bayesian asymptotics and the Jeffreys factor}
\label{subsec:bayes}

Let the family of per-antenna output distributions
$\bigl\{\{p(y\mid \xv)\}_y\colon \xv \in \Xc \bigr\}$ admit a one-to-one, smooth
parameterization by a vector $\thetav \in \Theta$:
\begin{equation}
  \bigl\{\{p(y\mid \xv)\}_y\colon\, \xv \in \Xc \bigr\}
  =
  \bigl\{\{p_{\thetav}(y)\}_y\colon\, \thetav \in \Theta \bigr\}, \label{eq:parameterization}
\end{equation}
where $\Theta \subset \mathbb{R}^d$ is the parameter space with dimension $d := \dim(\Theta)$. 
Moreover, let us define the  \emph{cost function}
$c(\thetav):\;\Theta\to \mathbb{R}_+$ such that the average-power constraint $\E[\|\xv\|^2] \le P$ is equivalent to
\begin{equation}
  \E_{\thetav}[\,c(\thetav)\,]\le P.
\end{equation}
Define the Fisher information matrix of $p_{\thetav}(y)$ 
\begin{equation}
  \pmb{J}(\thetav)
  :=\E_{p_{\thetav}}
\Bigl[
\nabla_{\thetav} \ln p_{\thetav}({y})\,
\nabla_{\thetav} \ln p_{\thetav}({y})^{\T}
\Bigr].  
\label{eq:Fisher}
\end{equation}%

Our main result extends Clarke and Barron’s theorem~\cite{CB94} to account for
both peak- and average-power constraints.
\begin{theorem}
  \label{thm:main}
  Consider a large-scale MIMO setting where the number of receive antennas
  $\nr$ grows large. Under mild regularity conditions~(see Appendix~\ref{app:CB_cond}), the
  channel capacity subject to both peak- and average-power constraints satisfies
  \begin{align}
    {C(P) = 
    \frac{d}{2}\log\frac{\nr}{2\pi e}} + \log \JF(\lambda^*) + o(1),
\label{eq:capa1}
  \end{align}%
  where \revision{the $o(1)$ term vanishes when $\nr \to \infty$;} $\lambda^*\ge0$ is the smallest $\lambda$ such that the \emph{(tilted)
Jeffreys prior} 
\begin{equation}
  w_{J,\lambda}(\thetav) = \frac{ 2^{-\lambda c(\thetav)} \sqrt{\det\pmb{J}(\thetav)}}{\int_{\Theta}
  2^{-\lambda c(\bm{\vartheta})} \sqrt{\det\pmb{J}(\bm{\vartheta})}\,d\bm{\vartheta}},\quad
  \thetav\in \Theta
  \label{eq:optimal-prior}
\end{equation}
satisfies the constraint $\E_{\thetav\sim w_{J,\lambda}}[\,c(\thetav)\,]\le P$;
$w_{J,\lambda^*}$ is the pdf of the asymptotically optimal input distribution;
we define 
\begin{equation}
  \mathsf{JF}(\lambda^*) :=
  {\int_{\Theta} 2^{-\lambda^* (c(\thetav)-P)}
  \sqrt{\det \pmb{J}(\thetav)} \,d\thetav}  \label{eq:JF}
\end{equation}%
and call it the \emph{(tilted) Jeffreys factor}. Moreover, for any input distribution $w$ that is strictly positive and continuous over $\Theta$, we have
\begin{align}
  I(\thetav; \yv) &= C(P) - D(w \| w_{J,\lambda^*}) \nonumber \\
  &\phantom{=}\ +  \lambda^*\, \E_{\thetav\sim w}\bigl[c(\thetav) - P \bigr]+ o(1),  \label{eq:C(w)}
\end{align}%
when $\nr\to\infty$. 
\end{theorem}
\begin{proof}
  See Appendix~\ref{app:th1}. 
\end{proof}
When the average-power constraint is not active (i.e., $P$ is large enough), 
we have $\lambda^*=0$ so that we recover Clarke and Barron's original result in~\cite{CB94}.

\subsection{Recipe for asymptotic capacity in large-scale MIMO}
\label{subsec:recipe}

From \eqref{eq:capa1}, we see that the asymptotic capacity only depends
on the dimension of the parameter space and the Fisher information
matrix of the per-antenna output distribution. Indeed, the first term $\frac{d}{2}\log\frac{\nr}{2\pi e}$ 
in \eqref{eq:capa1} is dominant and dictated by the number of \emph{degrees of
freedom} in the channel, while the second term shows the impact of the
actual channel distribution and the input constraints. Remarkably, the
Fisher information plays a central role in the capacity of large-scale
MIMO systems. The communication problem in the large receive antenna array regime essentially becomes a parameter estimation problem.

From \eqref{eq:capa1}, the asymptotic behavior of the capacity depends only on (i)~the dimension
$d$ of the parameter space and (ii)~the Fisher information matrix of the
single-letter output. Indeed, the dominant term
$\tfrac{d}{2}\log(\tfrac{\nr}{2\pi e})$ reflects the effective \emph{degrees of
freedom}, while the Jeffreys factor translates the channel distribution and
power constraints into a concise integral involving $\sqrt{\det\Jm(\thetav)}$.
In effect, communication in the large-$\nr$ regime becomes a parameter estimation
problem governed by the single-antenna Fisher information.

Hence, to compute the asymptotic capacity of a large-scale MIMO channel, one
proceeds as follows:
\begin{enumerate}
  \item Identify a one-to-one, smooth\footnote{\comment{When we say $\thetav\mapsto \{p_{\thetav}(y)\}$ is one-to-one and smooth, we mean that different values of $\theta$ produce different distributions, and for almost every $y$, the value of the density is smooth with respect to $\theta$, as specified by the regularity conditions in Appendix.}} mapping $\thetav\mapsto \{p_{\thetav}(y)\}_y$, such that \eqref{eq:parameterization} holds. 
\item Derive the Fisher information $\Jm(\thetav)$ from \eqref{eq:Fisher}.
\item Obtain the Jeffreys factor by evaluating the integral in
  \eqref{eq:JF}, after identifying the optimal tilting $\lambda^*$ that enforces
  $\E_{\thetav\sim w_{J,\lambda^*}}[\,c(\thetav)\,]\le P$.
\end{enumerate}

\comment{The optimal tilting $\lambda^*$ can be computed efficiently. Define 
\begin{align}
  \SM(\lambda) &:= \E_{\thetav\sim w_{J,\lambda}}[\,c(\thetav)\,]\\ 
  &= \frac{ \int_{\Theta} c(\bm{\vartheta})  2^{-\lambda c(\bm{\vartheta})} \sqrt{\det\pmb{J}(\bm{\vartheta})} d\bm{\vartheta}}{\int_{\Theta} 2^{-\lambda c(\bm{\vartheta})} \sqrt{\det\pmb{J}(\bm{\vartheta})}\,d\bm{\vartheta}}. 
\end{align}%
If $\SM(0) \le P$, then $\lambda^* = 0$ by definition. If $\SM(0) > P$, the following result guarantees that one can find $\lambda^*$ rapidly using \revision{bisection}.

  \begin{lemma}
    \label{lemma:M_decrease}
    If the cost function $\thetav \mapsto c(\thetav)$ is not constant, the mapping $\lambda \mapsto \SM(\lambda)$ is strictly decreasing. 
  \end{lemma}
  \begin{proof}
    Differentiating, we have
    \begin{align}
      \SM'(\lambda) &= -(\ln 2 ) \frac{ \int_{\Theta} c(\bm{\vartheta})^2  2^{-\lambda c(\bm{\vartheta})} \sqrt{\det\pmb{J}(\bm{\vartheta})}\, d\bm{\vartheta}}{\int_{\Theta} 2^{-\lambda c(\bm{\vartheta})} \sqrt{\det\pmb{J}(\bm{\vartheta})}\,d\bm{\vartheta}}
  \nonumber\\
  &\phantom{=}\ + (\ln 2) \left( \frac{ \int_{\Theta} c(\bm{\vartheta})  2^{-\lambda c(\bm{\vartheta})} \sqrt{\det\pmb{J}(\bm{\vartheta})} \,d\bm{\vartheta}}{\int_{\Theta} 2^{-\lambda c(\bm{\vartheta})} \sqrt{\det\pmb{J}(\bm{\vartheta})}\,d\bm{\vartheta}}
      \right)^2	\\	     
		    &= - \mathsf{Var}_{\theta \sim w_{J,\lambda}}(c(\theta))\ln 2 < 0, 
                  \end{align}%
                  where the variance \revision{is nonzero unless $\thetav\mapsto c(\thetav)$ is constant over the support $\Theta$ of $w_{J,\lambda}$}.
  \end{proof}
 
}

\subsection{A toy application: Real AWGN}
\label{subsec:toyAWGN}
To demonstrate the three-step recipe, let us consider the channel with additive white Gaussian noise~(AWGN):
\begin{equation}
y_i = x + z_i, \quad i\in[\nr], 
\end{equation}
\revision{where $z_i \sim \mathcal{N}(0,1)$, $i\in[\nr]$, are i.i.d.~Gaussian noise.} The input $x$
lies in $\Xc:=[-A,A]$, thus enforcing a peak-power constraint $x^2 \le A^2$, and
an average-power constraint
$  \E\bigl[x^2\bigr]\le P. $
The per-antenna conditional output distribution is Gaussian with mean
$x$ and variance~$1$, i.e., 
\begin{equation}
  p(y\mid x) = \phi(y-x), \quad x\in\Xc,\, y\in \mathbb{R}.  
\end{equation}%
It is straightforward to see that 
$x\mapsto \{\phi(y-x)\}_y$ is one-to-one and smooth. So we let 
$\theta = x$ and $\Theta = \Xc$, with the corresponding Fisher information
\begin{align}
  J(\theta) &= \E_y \left[ \Bigl({\partial \over \partial \theta} \ln \phi(y-\theta)\Bigr)^2 \right] \\
   &= \E_y \Bigl[ (y-\theta)^2 \Bigr] \\
   &= 1, 
\end{align}%
and cost function
$  c(\theta) := x^2 = \theta^2, $
for $\theta \in \Theta$.  
Using \eqref{eq:JF}, the Jeffreys factor is
\begin{align}
 \JF(\lambda) &=  \int_{-A}^{A} 2^{-\lambda (\theta^2-P)}
 \sqrt{J(\theta)} \,d \theta \label{eq:JF_SIMO}\\ 
 &= 2^{\lambda P}\sqrt{\frac{\pi}{\lambda\,\ln 2}}
\Bigl[ 1 -2Q\bigl(\sqrt{2\lambda \ln 2}A\bigr)
\Bigr],  
\end{align}%
and corresponding Jeffreys prior is 
\begin{equation}
  w_{J,\lambda}(\theta) = {1\over \JF(\lambda)} 2^{-\lambda (\theta^2-P)}
  \,\ind\left( \theta\in [-A, A] \right), \label{eq:JP0}
\end{equation}%
which is simply a truncated Gaussian. We can compute
\begin{equation}
  \SM(\lambda) = \E [c(\theta)] = {1\over \lambda \ln2} \left( {1\over2} - {A \,
  2^{\lambda(P-A^2)} \over \JF(\lambda)} \right), 
\end{equation}%
and solve for the smallest $\lambda=\lambda^*$ satisfying
$\SM(\lambda^*)\le P$.
The resulting input distribution $w_{J,\lambda^*}(\theta)$
is thus capacity-achieving in the large-$\nr$ regime.

\noindent\emph{Peak- vs.\ average-power constraints.}
When $\lambda=0$, the Jeffreys prior $w_{J,0}$ becomes uniform on $[-A,A]$, with second moment
$A^2/3$. Hence if $P \ge A^2/3$, the average-power constraint is inactive, and
$\lambda^*=0$. From \eqref{eq:capa1}, the corresponding capacity simplifies to
\begin{equation}
C =
\frac{1}{2}\,\log\frac{2\nr A^2}{\pi e} + o(1), 
\label{eq:capa-1D-AWGN}
\end{equation}
achieved by a uniform input. 
In general, for $P < A^2/3$, both constraints become active, leading to the
asymptotic capacity expression
\begin{equation}
C(P) =
\frac{1}{2}\,\log\frac{\nr \, \JF(\lambda^*)^2}{2\pi e} + o(1), 
\label{eq:SIMO}
\end{equation}
where $\JF(\lambda^*) \le 2A$ serves as an effective amplitude measure. In the
small-$P$ regime, the peak constraint is almost never saturated; hence $x$ behaves
like a Gaussian with variance approximately $1/(2\lambda^*\ln 2)$,
giving $P \approx 1/(2\lambda^*\ln 2)$ and
\begin{equation}
  \JF(\lambda^*) \approx 2^{\lambda^*\! P} \sqrt{2 \pi P} \approx
  \sqrt{2\pi e P}.
\end{equation}%
Thus, $C(P)$ reduces to $\tfrac12\log(\nr P) + o(1)$, aligning with the exact
capacity without peak power constraint $\tfrac12\log(1+\nr P)\approx \tfrac12\log(\nr P)$ for large~$\nr$.
Figure~\ref{fig:toy-example} illustrates the Jeffreys factor $\JF(\lambda^*)$ vs.\ $\sqrt{P}$, comparing it to the approximation for small $P$.

\begin{figure}[t]
\centering
\includegraphics[width=\columnwidth]{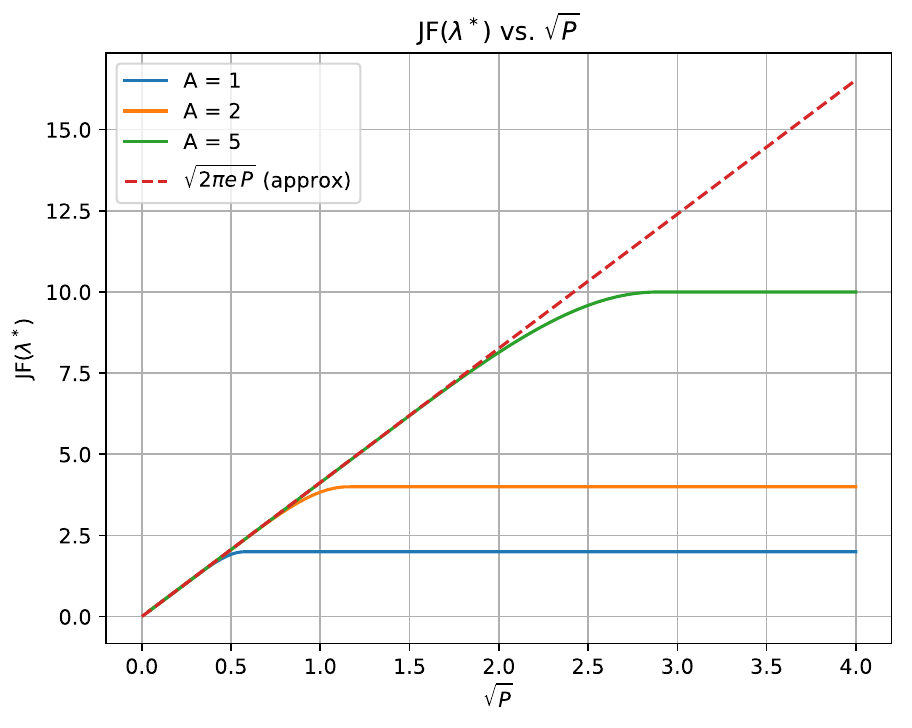}
\caption{Illustration of the Jeffreys factor in the SIMO real-AWGN case vs.\ $\lambda$.}
\label{fig:toy-example}
\end{figure}

\section{Applications to Different Channels}
\label{sec:applications}

We have illustrated our method on a toy example where the capacity is well understood. In this section, we show how the proposed framework extends to a range of channels for which the exact capacity is difficult to characterize, especially in the presence of both peak- and average-power constraints.

\subsection{AWGN with clipping}
\label{subsec:awgnclip}
Consider the same real AWGN channel as in the toy example, but with output clipping, i.e.,
$y_i = \clip(x+z_i)$  with $\clip(u):=\max(-B, \min(u, B))$. 
This is a simplified model of output saturation at the receiver. 
Hence, the conditional distribution $p(y \mid x)$ is a truncated Gaussian plus point masses at $\pm B$, i.e., 
\begin{align}
 p(y \mid x) &= {\phi(y-x)}  \mathbf{1}(|y|<B)   
 \nonumber \\ &\phantom{=} \ + {Q(B-x)}\delta(y-B)  +
 {Q(B+x)}\delta(y+B),  
\end{align}%
where $\delta(\cdot)$ is the Dirac function\footnote{\comment{Formally, the density is defined (and absolutely continuous) with respect to the measure $dy + \delta_{+B}(dy) + \delta_{-B}(dy)$, i.e., the sum of the Lebesgue measure on $\mathbb{R}$ and point masses at $+B$ and $-B$, so that for any measurable function $f$, the conditional expectation can be written as $\mathbb{E}\bigl[f(y) \mid x\bigr] = f(B) Q(B-x) +  f(-B) Q(B+x) + \int_{-\infty}^{+\infty} f(y) \phi(y-x) \ind(|y| \le B) dy
$.} }. Again, 
$x \mapsto \{ p(y\mid x) \}_y$ is one-to-one and smooth on $[-A, A]$. Thus, we let $\theta = x$ and
$\Theta = \Xc$ as in the toy example. 
We can now explicitly compute the Fisher information:
\begin{equation}
J(\theta)
 := 1 - \sum_{s\in\{B+\theta,\,B-\theta\}} 
 \biggl(Q(s)+s\,\phi(s) - \frac{ \phi^2(s) }{Q(s)} \biggr),
\label{eq:clipping-FI}
\end{equation}
which is strictly less than 1 for most $\theta$ due to the clipping, 
manifested by the negative term\footnote{\comment{Indeed, since its derivative is 
$-{\phi(s)\over Q^2(s)}\left( s Q(s) - \phi(s) \right)^2 \le 0$, the function inside the sum is decreasing with $s$ and decays to $0$ when $s\to\infty$.}}. 

Plugging $J(\theta)$ into \eqref{eq:JF_SIMO} yields the Jeffreys factor and Jeffreys prior \eqref{eq:JP0}. Notably, for small $B$, the prior peaks near $\theta=0$ because strong clipping reduces the effective
input range. Figure~\ref{fig:JP-clipping} shows this numerically for different
clipping thresholds $B=\alpha A$ when $P=A^2/9$. In the right plot, we see how the clipping ratio affects proportionally the ratio $\JF(\lambda^*)/A$. 

\begin{figure}[t]
\centering
\includegraphics[width=\columnwidth]{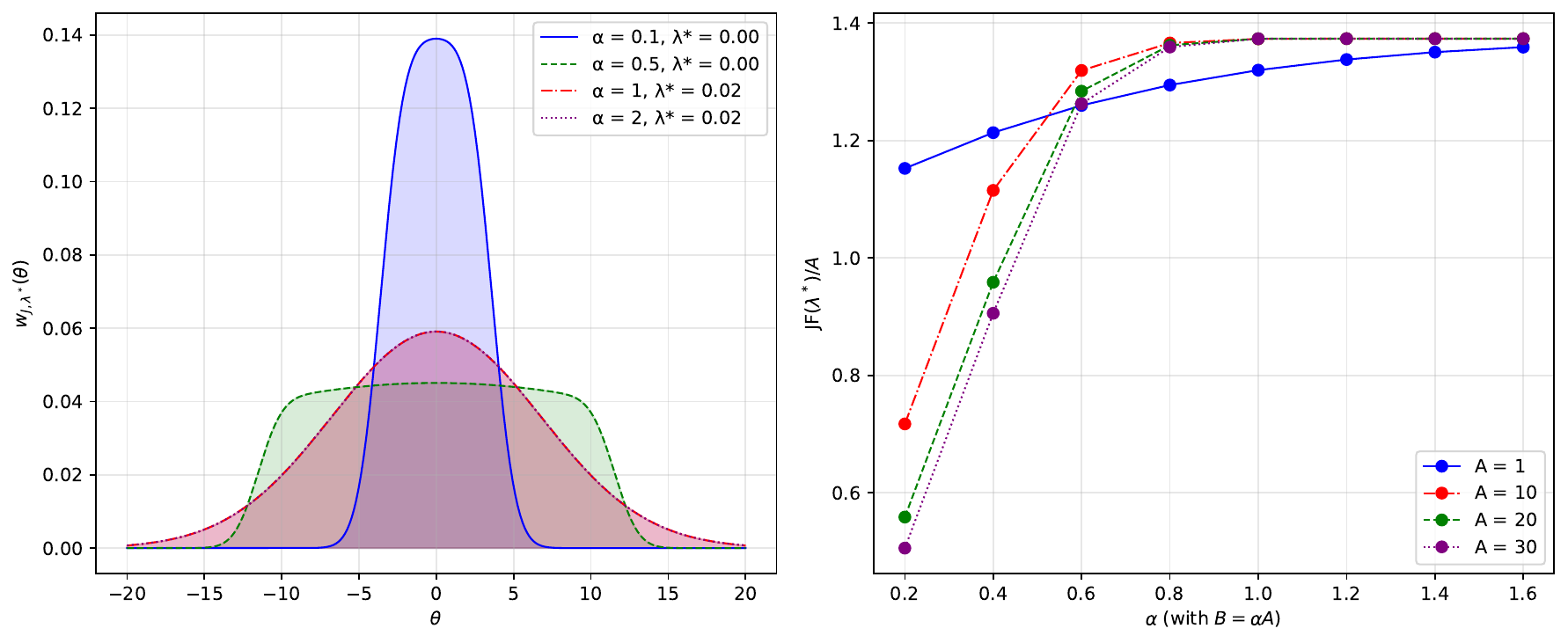}
\caption{Jeffreys prior for a clipped SIMO AWGN channel with different
clipping thresholds $B=\alpha A$. Average power $P=A^2/9$. Left: Jeffreys prior with different clipping ratios. Right: Jeffreys factor normalized by the peak amplitude $A$ vs.~clipping ratio~$\alpha$. }
\label{fig:JP-clipping}
\end{figure}

\subsection{AWGN with low-resolution ADC}
\label{subsec:awgnquant}

Continuing in the real SIMO setting, suppose each output is quantized into $L$
levels via thresholds 
\[
  -\infty = t_0 < t_1 < \dots < t_{L-1} < t_L = \infty.
\]
Hence, the $i$-th antenna output is
\begin{equation}
{y_i} = \sum_{\ell=1}^{L} \ell\;\ind({x}+z_i\in(t_{\ell-1},t_\ell]),
\quad i\in[\nr]. 
\end{equation}%
The per-output conditional distribution is discrete in this case
\begin{equation}
  p(y_i \mid {x}) = Q({x}-t_{y_i}) - Q({x}-t_{y_i-1}),\quad y_i\in[L]. \label{eq:tmp1211}
\end{equation}%
Again, we can verify that
$x\mapsto \{p(y\mid x)\}_y$ is one-to-one and smooth, so we set $\theta = x$ and
$\Theta = \Xc$. 
Computing the Fisher information is straightforward:
\begin{align}
J(\theta) &= \sum_{\ell=1}^L \frac{\bigl[\phi({x}-t_{\ell-1}) - \phi({x}-t_{\ell})\bigr]^2} {Q({x}-t_{\ell}) - Q({x}-t_{\ell-1})}.
\label{eq:lowres-FI}
\end{align}
When $L=2$ with $t_1=0$, we recover the 1-bit ADC~(sign) case of \cite{YC2024}.

Figure~\ref{fig:JP-ADC} illustrates the Jeffreys prior and factor for
various $L$ under a peak constraint $|x|\le A$ and average power $P=A^2/9$.
The quantizers are uniform with $L$ levels and clipping at $\pm A$ for $L > 2$. For instance, when $L=4$, we have
$t_1 = -A$, $t_2 = 0$, $t_3 = A$.
For a small peak amplitude $A$, the Jeffreys factor saturates quickly with $L$, suggesting that moving beyond coarse quantization yields diminishing gains when $A$ is not large. 

\revision{In a complex-valued model where both the real and imaginary components are quantized with $L$ levels, each output can be represented as a pair in $[L]^2$. The conditional output distribution in~\eqref{eq:tmp1211} is then replaced by the product of the conditional distributions of the two components. The corresponding Fisher information can be derived in an analogous, albeit more involved, manner.}

\begin{figure}[t]
\centering
\includegraphics[width=\columnwidth]{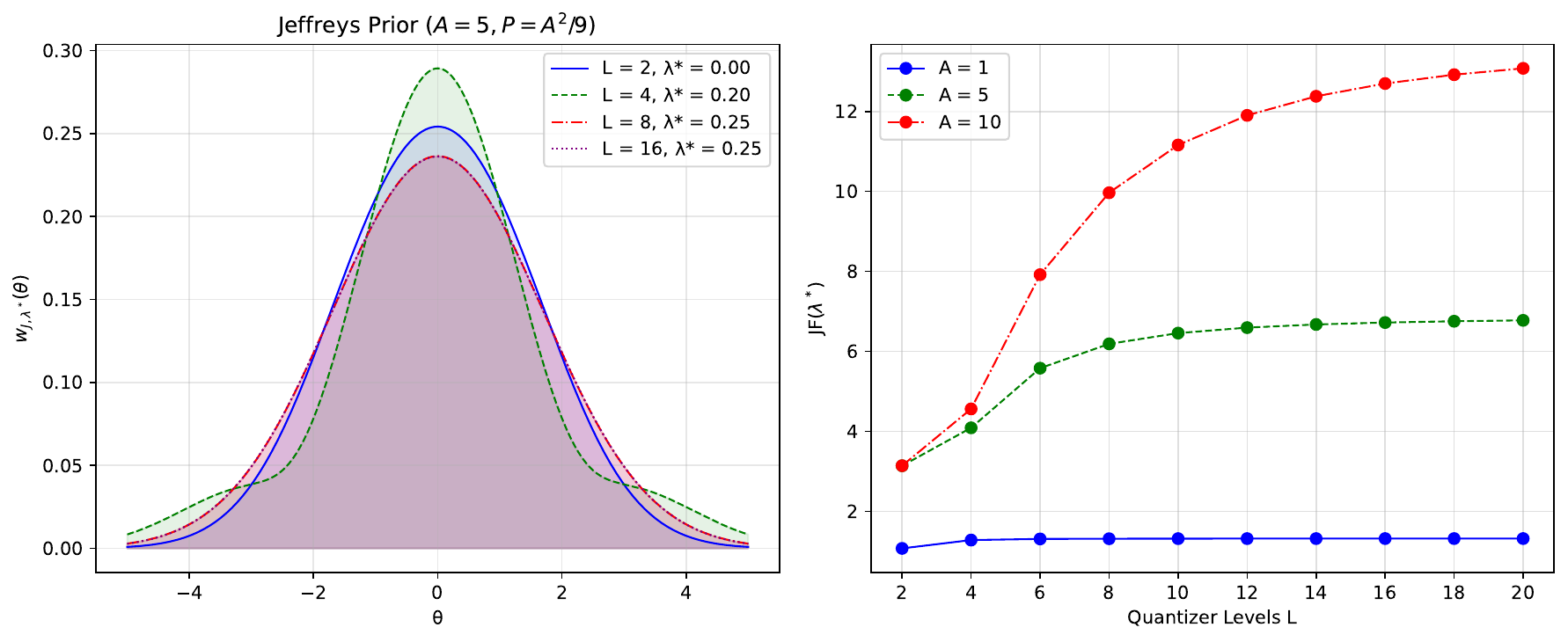}
\caption{Jeffreys prior (left) and Jeffreys factor (right) for real SIMO
AWGN with $L$-level quantizers and peak amplitude $A$. Average power constraint $P = A^2/9$. }
\label{fig:JP-ADC}
\end{figure}

\subsection{Complex AWGN with \revision{energy detection}}
\label{subsec:phaseNoise}

\revision{Next, consider the complex AWGN SIMO channel in which the receiver has no knowledge
of the channel phase, for instance due to uniform phase noise.
In this scenario, only the \emph{magnitude} of each received signal,
$|y_i|$, $i\in[\nr]$, is observable.}
Equivalently, we consider the following model
\begin{equation}
|y_i| =  |x + z_i|, \quad i \in [\nr],
\end{equation}%
where $z_i\sim \mathcal{CN}(0, 1)$, and  
$$ x\in \Xc:= \bigl\{x \in \mathbb{C}:\, |x|\le A \bigr\}.$$  
Then, conditioned on $x$, the random variable $\tilde{y}_i :=
2|y_i|^2$ follows a noncentral
$\chi^2$ distribution with $2$ degrees of freedom and noncentrality parameter
$2|x|^2$. In particular, 
\begin{equation}
p(\tilde{y} \mid x)
=\tfrac12\exp\bigl(-\tfrac{\tilde{y}+2|x|^2}{2}\bigr)
I_0\bigl(\sqrt{2|x|^2\,\tilde{y}}\bigr),\quad \tilde{y}\ge0,
\end{equation}%
where $I_k(\cdot)$, $k=0,1,\ldots$, is the order-$k$ modified Bessel
function of the first kind. Since the distribution $p(\tilde{y}
\mid x)$ depends on $x$ only through $|x|$, we let $\theta = |x|$ and $\Theta = [0, A]$, and we can check that the
parameterization $\theta \mapsto \{p(y \mid |x|=\theta)\}_y$ is 
one-to-one and smooth. The cost function is $c(\theta) = \theta^2 =
|x|^2$. 

It follows that the Fisher information has the following expression
\begin{equation}
J(\theta)= \mathbb{E}_{\tilde{y}|\theta}\!\left[\left(-{2\theta} + \sqrt{{2\tilde{y}}}\, \frac{I_1\Bigl(\theta\sqrt{{2\tilde{y}}}\Bigr)} {I_0\Bigl(\theta\sqrt{{2\tilde{y}}}\Bigr)}\right)^2\right]
\end{equation}%
that can be computed numerically. 

Since the problem depends only on the input magnitude, it suffices to restrict the input be real and positive, i.e., $x\ge0$. 
The resulting Jeffreys prior and factor, under average power $P=A^2/9$, is shown in
Figure~\ref{fig:JP-PN}. We observe that the complex channel with
\revision{energy detection} has a similar behavior as the real SIMO channel, 
for the phase information is lost.  

\begin{figure}[t]
\centering
\includegraphics[width=\columnwidth]{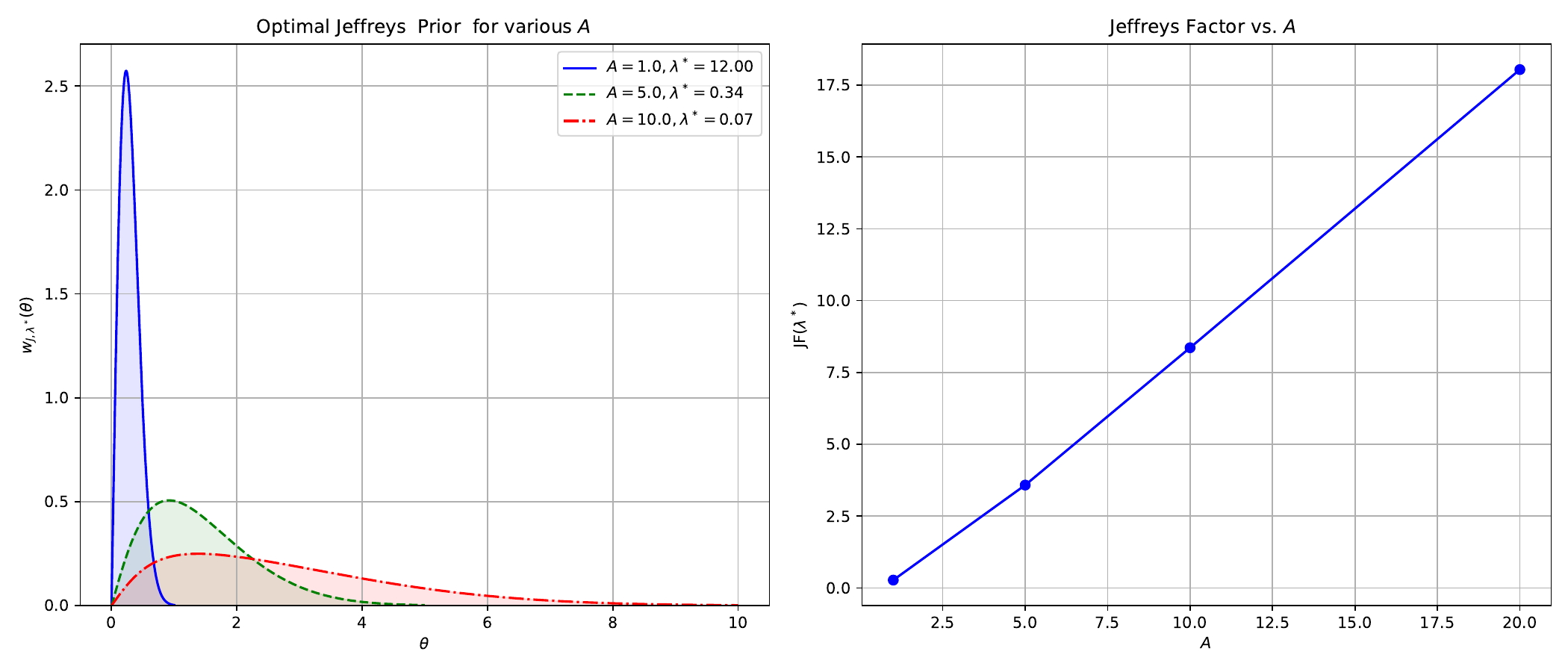}
\caption{Jeffreys prior~(left) and Jeffreys factor~(right) for a complex SIMO channel with \revision{energy detection}.}
\label{fig:JP-PN}
\end{figure}

\subsection{MIMO channels with imperfect CSI}
\label{sec:app_MIMO}

In the following, we demonstrate how the proposed framework applies in scenarios with
 fading or other forms of \emph{channel states}. In particular, we
 focus on settings in which the channel states are i.i.d.\ across antennas,
 known to the receiver but not to the transmitter.

Consider the conventional MIMO fading model~\eqref{eq:MIMO}, but assume the
receiver obtains only a partial estimate $\hat{\Hm}$ of the channel matrix
$\Hm$. Specifically, let
\begin{equation}
 \Hm = \hat{\Hm} + \tilde{\Hm},
\end{equation}%
where $\tilde{\Hm}$ is estimation noise independent of $\hat{\Hm}$. The input
$\xv$ is subject to both a peak-power constraint
\begin{equation}
  \xv \in \Xc := \{\xv \in \mathbb{C}^{\nt}:\, \|\xv\| \le A \},
\end{equation}%
and an average power constraint $\E [\|\xv\|^2] \le P$. 

The rows of $\hat{\Hm}$ are i.i.d.~$\sim p(\hat{\hv})$, while the entries of
$\tilde{\Hm}$ are i.i.d.~$\sim\mathcal{CN}(0,\sigma^2)$. This model recovers the
{noncoherent} channel when $\hat{\Hm}=0$ and the {coherent} channel
when $\sigma=0$. The joint pdf of the outputs given $\xv$ is
\begin{equation}
  p(\yv, \hat{\Hm} \mid \xv) = \prod_{i=1}^{\nr} p(y_i, \hat{\hv}_i \mid \xv),
\end{equation}%
where the per-output distribution is 
\begin{equation}
  p(y, \hat{\hv} \mid \xv) = p(\hat{\hv}) \, \revision{\phi_c\!\left({y - 
  \hat{\hv}^{\T} \xv; {1+\sigma^2 \|\xv\|^2}} \right).} \label{eq:tmp2111}
\end{equation}%

\revision{Under some regularity conditions on $p(\hat{\hv})$~(e.g., not deterministic or degenerate), 
the mapping $\xv\mapsto \{p(y, \hat{\hv} \mid \xv)\}_{y,\hat{\hv}}$ is a
one-to-one smooth parameterization. In this case, let us set 
$\thetav = [\Re(\xv); \Im(\xv)]$ and 
$\Theta =\{\thetav\in\mathbb{R}^{2\nt\times 1}:\,\|\thetav\|\le A\}.$}
The cost function becomes $c(\thetav) = \|\thetav\|^2$. 
The Fisher information matrix w.r.t.~$\thetav$ is
\begin{IEEEeqnarray}{rCl}
  \Jm(\thetav) = \frac{2}{1 + \sigma^2
  \|\thetav\|^2} \, \GR + \frac{4
  \sigma^4}{\left( 1 + \sigma^2 \|\thetav\|^2 \right )^2} \,
  \thetav \thetav^{\T},\nonumber
\end{IEEEeqnarray}%
where $\GR$ is the real lifting\footnote{We say that $\left[
\begin{smallmatrix} \Re(\Mm)&
  -\Im(\Mm)\\ \Im(\Mm)& \Re(\Mm) \end{smallmatrix} \right]$ is the real
  lifting of $\Mm$. } of $\E\bigl[\hat{\hv} \hat{\hv}^\T \bigr]$. 
If $\GR$ is invertible, one may
further compute
\begin{align}
  \det\left(\Jm(\thetav)\right) &=
  \left(\frac{2}{1 + \sigma^2
  \|\thetav\|^2} \right)^{2\nt}  \det(\GR) 
   \left( 1 + {2
  \sigma^4 \thetav^\T \GR^{-1} \thetav \over 1 + \sigma^2
  \|\thetav\|^2}   \right). 
\end{align}

As a simple example, assume $\hat{\hv}$ is zero-mean with uncorrelated real and
imaginary parts of variance $(1-\sigma^2)/2$. Then $\GR$ is a scalar multiple of the
identity, implying both $\Jm(\thetav)$ and $c(\thetav)$ depend only on
$\|\thetav\|$. In that case, the {Jeffreys prior} is \emph{isotropic} in
$\mathbb{R}^{2\nt}$. One can write $\thetav = r\,\uv$, where $r\ge 0$ is the
radius and $\uv \in \Sphere_{2\nt-1}$ is a uniform direction vector with $\|\uv\|=1$. 
The radial density becomes
\begin{align}
  p(r) &\propto  2^{-\lambda^* r^2} \left(\frac{2(1-\sigma^2)}{1 + \sigma^2
  r^2} \right)^{\nt}  \nonumber \\
  &\phantom{\propto} \ \cdot \sqrt{ 1 + \frac{2 \sigma^4}{1-\sigma^2} { r^2  \over 1 + \sigma^2
  r^2 } }  r^{2\nt-1}, \quad r\ge0. 
\end{align}%
From Figure~\ref{fig:JP-MIMO}, one observes that when the estimation error increases the average radius reduces to avoid excessive self-interference. The Jeffreys factor increases substantially with decreasing estimation error.

\begin{figure}[t]
\centering
\includegraphics[width=\columnwidth]{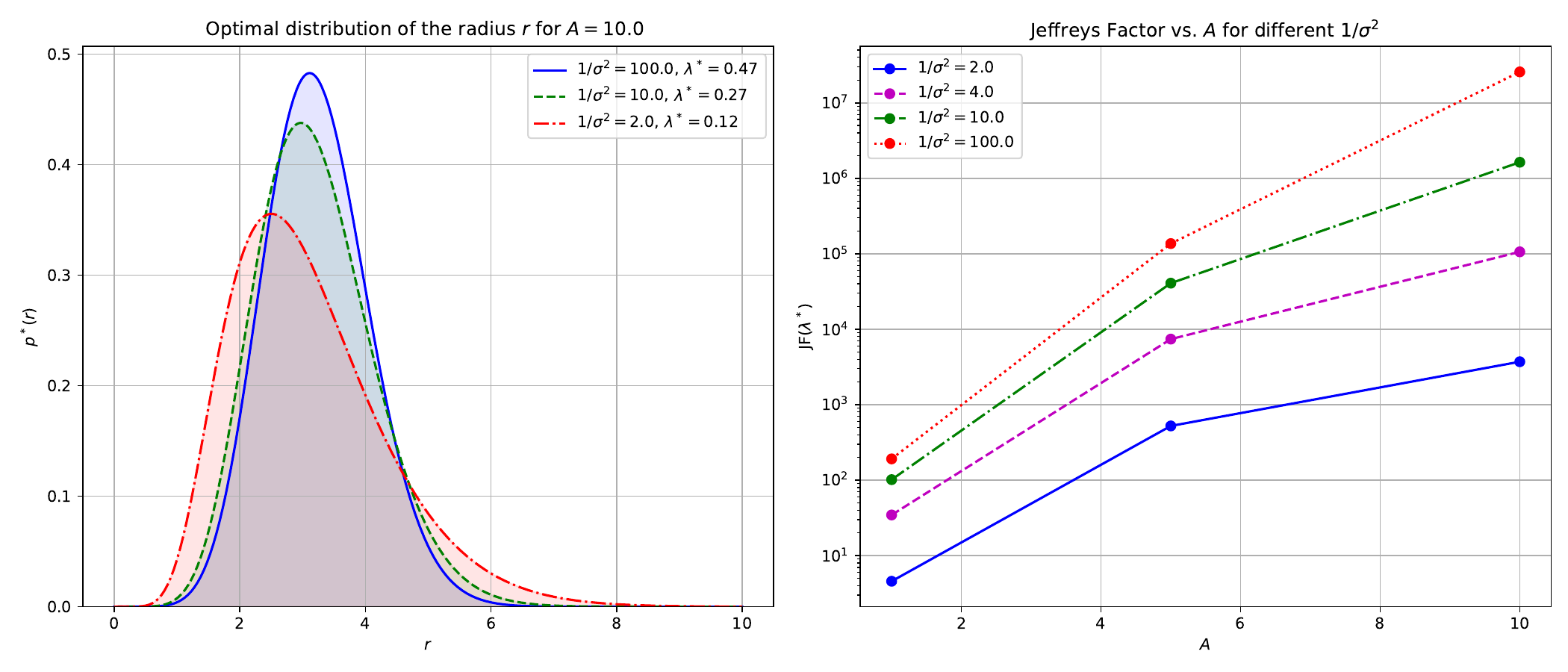}
\caption{MIMO channel with imperfect CSIR, $\nt=4$ transmit antennas, peak-power constraint $A$, average power constraint $P=A^2/9$. Left: optimal distribution on the input radius derived from the Jeffreys prior. Right: The Jeffreys factor.}
\label{fig:JP-MIMO}
\end{figure}

\revision{
In the extreme case $\hat{\Hm}=0$, corresponding to the {noncoherent channel},
the conditional distribution in~\eqref{eq:tmp2111} reduces to
\[
p(y\mid\xv)=\phi_c\big(y;\,1+\sigma^2\|\xv\|^2\big).
\]
In this case, the output distribution depends on the input only through its norm.
Defining $\theta=\|\xv\|$ with parameter space $\Theta=[0,A]$, the Fisher information
with respect to~$\theta$ is
$
J(\theta)=\frac{4\sigma^4\theta^2}{(1+\sigma^2\theta^2)^2},
$
which is independent of the number of transmit antennas.
The corresponding Jeffreys factor admits the following closed-form expression:
\begin{equation}
\JF(\lambda)
= 2^{\lambda/\sigma^2}
\!\left[
E_1\!\left(\tfrac{\lambda\ln 2}{\sigma^2}\right)
- E_1\!\left(\tfrac{\lambda\ln 2}{\sigma^2}\big(1+\sigma^2A^2\big)\right)
\right],
\end{equation}%
where $E_1(x):=\int_x^\infty \frac{e^{-t}}{t}\,dt$ denotes the exponential integral function.
}

\subsection{Optical SIMO Poisson channel}
\label{subsec:poisson-opt}

We next consider an \emph{optical intensity} channel, where the observed
signal is a Poisson random variable whose intensity depends on the
input and potential fading. Specifically, for the SIMO setup, each
output $y_i$ follows
\begin{equation}
y_i \sim \mathrm{Poisson}(\mu), \quad \mu = h_i\, x + \mu_{0,i},
\end{equation}
where $x\ge 0$ is the (scalar) transmitted intensity, $h_i\ge0$ is the fading
coefficient, and $\mu_{0,i}\ge0$ is the background intensity. The distribution
is thus
\begin{equation}
p(y_i \mid x, h_i, \mu_{0,i}) =
\frac{(h_i\,x+\mu_{0,i})^{y_i}}{y_i!}\exp\{-(h_i\,x+\mu_{0,i})\}, 
\end{equation}
where $y_i = 0, 1, \ldots$.  We assume that $(h_i,\mu_{0,i})$ are mutually independent and are i.i.d.\ across $i$ and known at the receiver but not at the transmitter. Thus,
we have
\begin{equation}
 p(y_i,h_i,\mu_{0,i} \mid x) = p(y_i\mid x, h_i, \mu_{0,i})\,
 p(h_i)\, p(\mu_{0,i}). 
\end{equation}
Again, the parameterization $x\mapsto \{p(y,h,\mu \mid x)\}_{y,h,\mu}$ of the
per-output distribution is one-to-one and smooth. We set $\theta = x$, 
$\Theta = [0, A]$, and $c(\theta) = \theta^2$.  
The Fisher information is 
\begin{equation}
J(\theta) = \E_{h,\mu}\biggl[\frac{h^2}{h\,\theta + \mu}\biggr].
\end{equation}

Without fading and background randomness, e.g., $h=1$, the 
Jeffreys prior becomes 
\begin{equation}
  w_{J,\lambda^*}(\theta) \propto \frac{2^{-\lambda^*
  \theta^2}}{\sqrt{\theta+\mu}},\quad \theta \in \Theta.
\end{equation}
Figure~\ref{fig:JP-optical} shows the resulting Jeffreys prior and Jeffreys
factor for different background intensities~$\mu = \alpha A$. One can observe how the Jeffreys factor decreases with increasing background intensity.

\begin{figure}[t]
\centering
\includegraphics[width=\columnwidth]{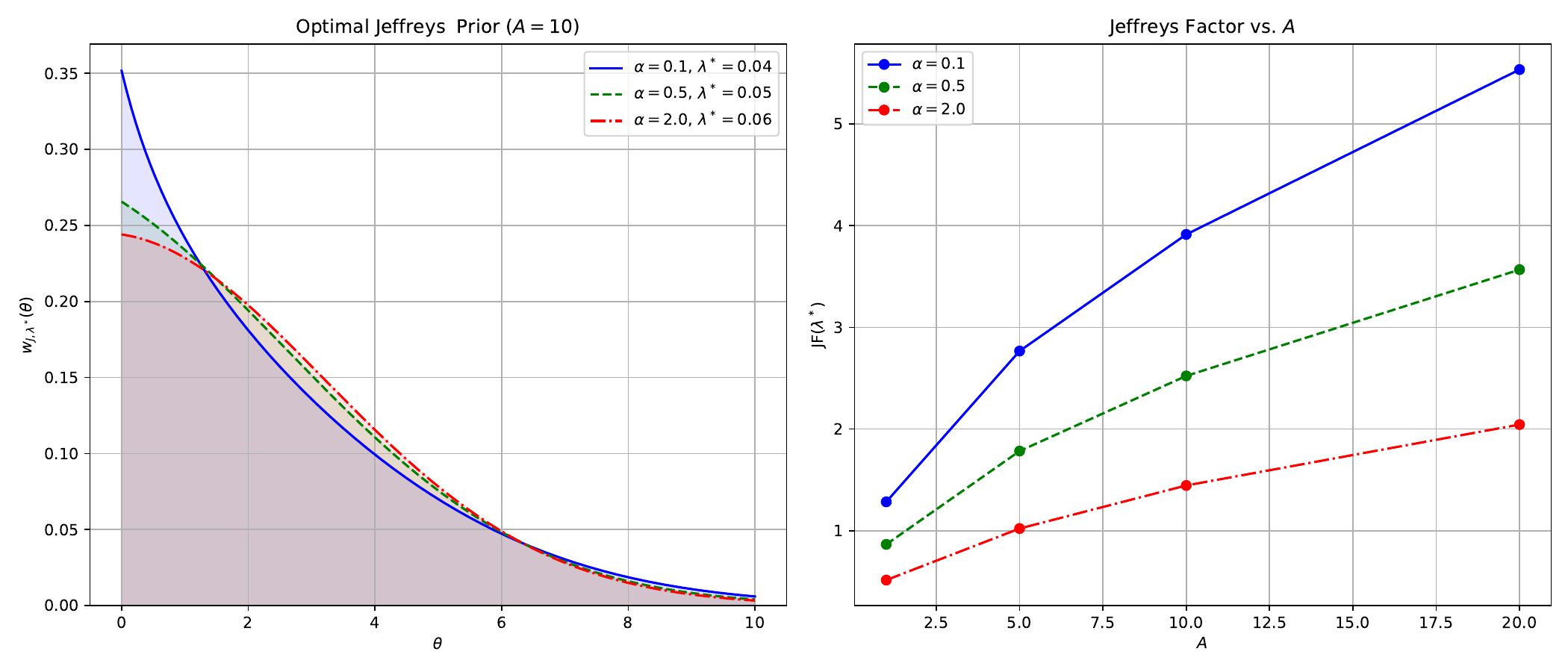}
\caption{Jeffreys prior (left) and Jeffreys factor (right) for a SIMO Poisson
intensity channel with peak intensity $A$, background intensity $\mu = \alpha A$ for different $\alpha$. Average power constraint $P=A^2/9$. }
\label{fig:JP-optical}
\end{figure}

\subsection{1-bit ADC with dithering}
\label{subsec:1bit-dithering}

In the two preceding examples, the channel state arises naturally~(e.g., from
fading or background processes). Here, we illustrate how \emph{artificial
dithering} can be introduced in a SIMO channel with 1-bit ADCs to
potentially improve capacity, an application that neatly fits the same
Bayesian asymptotic framework.

Suppose each output is 1-bit quantized as
\begin{equation}
y_i =\mathrm{sign}(x + z_i - s_i ), \quad i\in[\nr],
\end{equation}
where $z_i\sim\mathcal{N}(0,1)$, $i\in[\nr]$, are i.i.d.~noise and $s_i$ is an
i.i.d.~\emph{dithering} random variable, known at the receiver but not at the
transmitter. Conceptually, $s_i$ artificially ``shifts'' the observation,
spreading out the threshold across receive antennas. If $s_i=0$, this reduces
to the conventional 1-bit ADC channel. 

The per-antenna output distribution is effectively 
\begin{equation}
  p(y, s \mid x) = p(s)\;Q\bigl((s-x)\,y\bigr), \quad y\in\{+1,-1\},\,s\in\mathbb{R}.
\end{equation}
We can check that $x\mapsto \{p(y,s \mid x)\}_{y,s}$ is one-to-one and smooth.
Let $\theta = x$ and $\Theta = [-A, A]$, we can compute the Fisher
information 
\begin{equation}
J(\theta) =\E_{s}\!\biggl[ \frac{\phi^2(\theta-s)}{ Q(\theta-s)\bigl[1-Q(\theta-s)\bigr] } \biggr].
\end{equation}
Figure~\ref{fig:JP-dither} shows the resulting Jeffreys prior and
factor for a uniform dithering distribution $p(s)$ over $N$ equally spaced
points in $[-\alpha A, \alpha A]$. Even for $N=3$, dithering can significantly
increase the Jeffreys factor, especially at larger $A$.  Interestingly, the
optimal $\alpha$ does not exceed $1$. In practice, dithering need not be truly
random: one can partition the antennas into $N$ groups and apply a fixed shift
for each group, so long as it remains independent of $z_i$. This again
underscores how \emph{artificial channel states} can be incorporated into the
Bayesian asymptotic framework in which the potential performance gain may be
analyzed quantitatively. 

\newrevision{
Note that the beneficial effect of dithering in multi-antenna systems has been previously observed in the context of coarse quantization and generalized mutual information~(GMI) analysis (see, e.g., \cite{liang2016mixed, wang2022generalized}). In contrast, our formulation makes the impact of dithering explicit through the Fisher information and the resulting Jeffreys factor, thereby providing a simple and quantitative characterization of the associated capacity gain within the Bayesian asymptotic framework.
}

\begin{figure}[t]
\centering
\includegraphics[width=\columnwidth]{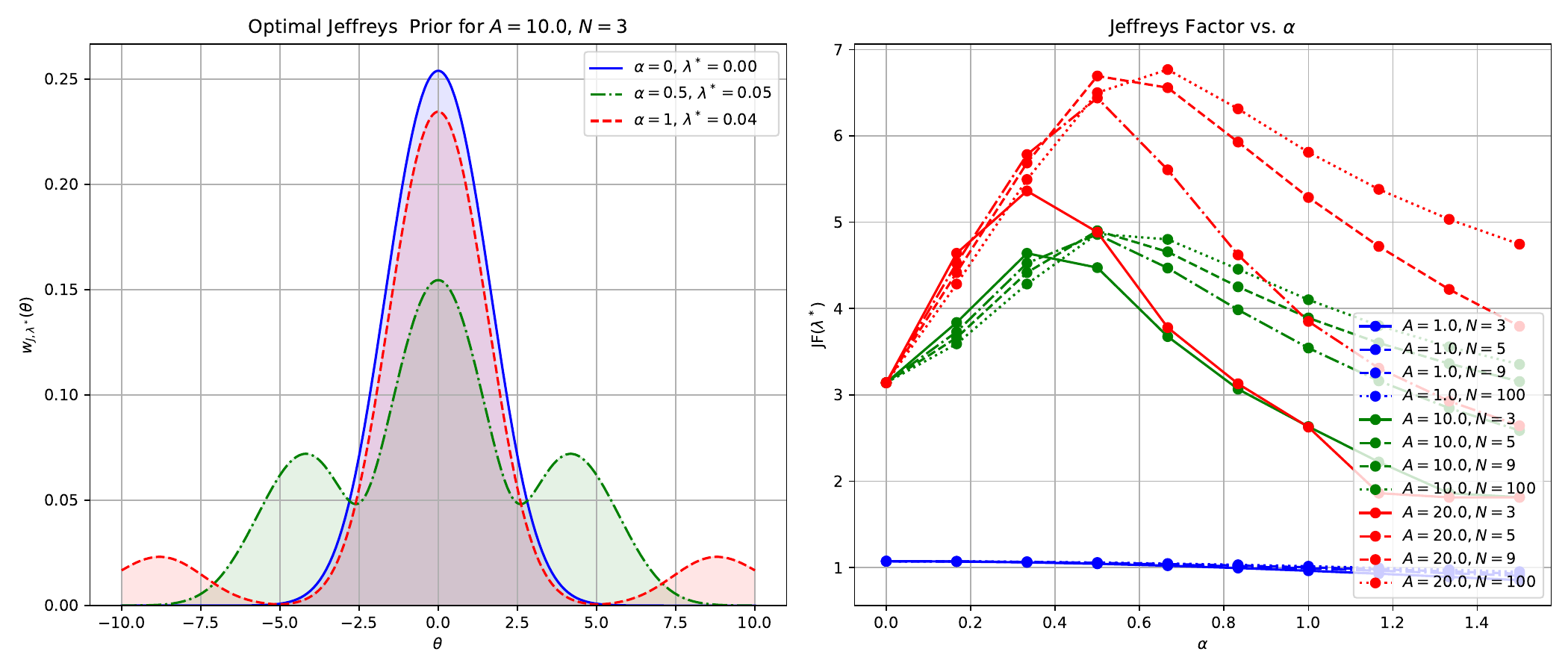}
\caption{SIMO channel with 1-bit ADCs and artificial dithering in $[-\alpha A,\,\alpha A]$. Peak amplitude $A$, average power constraint $P = A^2/9$. Left: Jeffreys prior for different values of $\alpha$. Right: The Jeffreys factor for different values of $A$ and $N$. Solid, dashed, dotted, and dash-dotted lines correspond to $N = 3, 9, 20$, and $100$, respectively.}
\label{fig:JP-dither}
\end{figure}

\section{Constellation Design}
\label{sec:constellation}

Although the {continuous} Jeffreys prior is asymptotically optimal, it may
be impractical to implement in many communication systems. Instead, a
{discrete} constellation~(finite input alphabet) is often desired. In this
section, we present a systematic approach to designing a constellation that
closely approximates the Jeffreys prior for scalar inputs. We will discuss how
this approach is conceptually similar to compander-based quantization, where a
nonlinear transform ``uniformizes'' the signal prior to quantization. 
\revision{Generalization to the vector inputs will also be discussed. }

\subsection{Jeffreys constellation}
Suppose the optimal (continuous) prior is $w_{J,\lambda^*}(\theta)$ over $\theta\in\Theta$,
with cumulative distribution function~(cdf)
\begin{equation}
  F(u) := \int_{-A}^u w_{J,\lambda^*}(x) d x. 
  \label{eq:cdf}
\end{equation}
Then one can construct $\theta$ as
\begin{equation}
  \theta = g(u) := {F}^{-1}(u), \quad u\in[0,1]. 
\end{equation}%
If $u$ is uniformly distributed on $[0,1]$, then $\theta=g(u)$ follows the
Jeffreys prior $w_{J,\lambda^*}(\theta)$. To discretize $\theta$, let $\mathcal{S}$ be a
finite set of equally spaced points in $[0,1]$. We then define the
\emph{Jeffreys constellation}
\begin{equation}
  \mathcal{X}_{J}
  :=
  \bigl\{\,c_P\,F^{-1}(u):\;u\in\mathcal{S}\bigr\},
\end{equation}
where $c_P\le 1$ is a scaling factor ensuring that the average-power constraint
$P$ is satisfied. Essentially, this approach applies a smooth,
nonlinear mapping to a uniform grid of $u$, yielding a finite set of input
points approximating the \emph{continuous} Jeffreys prior.

\begin{remark}
This design mirrors the idea of \emph{compander quantization} in classical
signal processing, where a nonlinear transform (the ``companding'' function)
re-maps the signal into a near-uniform distribution before applying a uniform
quantizer.  Here, the inverse transform $g^{-1}(\theta)$ is precisely the
Jeffreys prior's cdf, acting like a compander that \emph{uniformizes} the
optimal input distribution. In effect, it reparameterizes the conditional
channel distribution so that the Jeffreys prior in the new parameter space
becomes uniform.  
\end{remark}

Figure~\ref{fig:constellation_design}~(left) compares the performance of this
Jeffreys constellation with a standard pulse-amplitude modulation (PAM) grid
for a SIMO channel with 1-bit ADC. For each constellation, we compute the mutual information with uniform input distribution and an optimized input distribution~(via the Blahut–Arimoto algorithm). \revision{The Jeffreys constellation consistently outperforms the conventional PAM design, with a significant gain when the SNR is higher~($A=10$). 
Moreover, further distribution optimization yields only minor gains for the Jeffreys constellation.} 
This suggests that uniform probabilities over the Jeffreys constellation points are already near optimal.

\begin{figure*}[t]
\centering
\includegraphics[width=2\columnwidth]{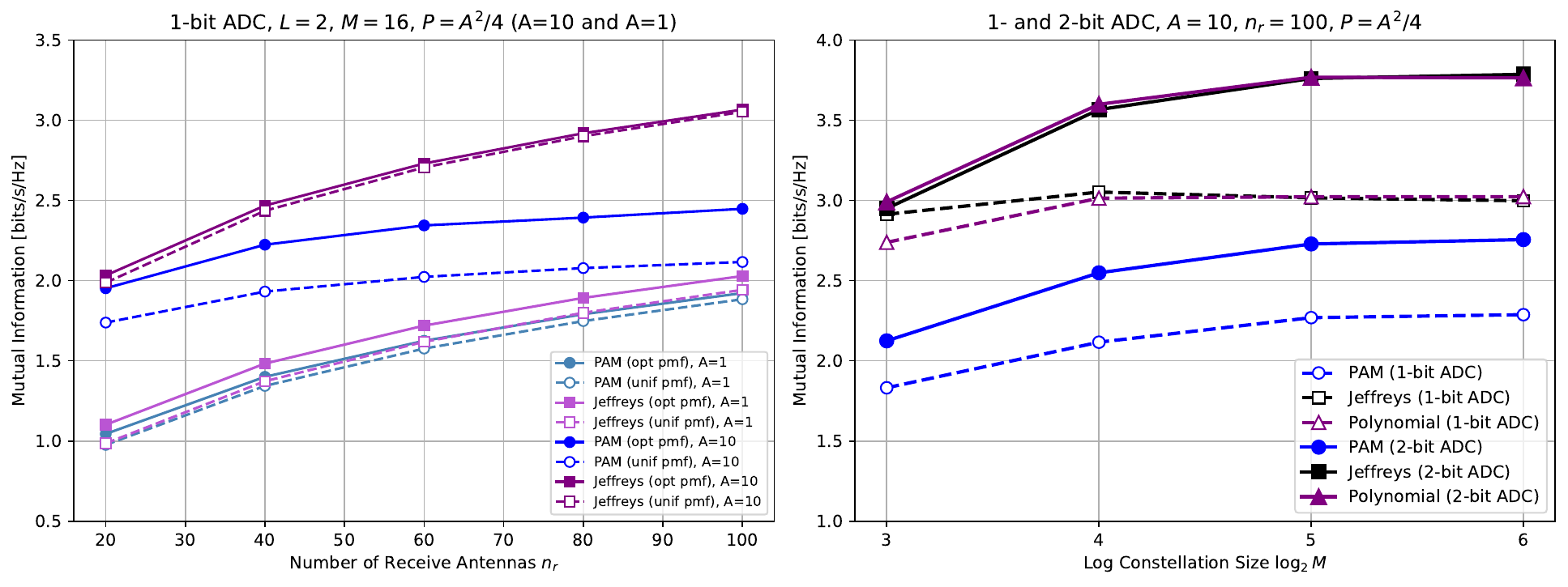}
\caption{\revision{Constellation design for a SIMO channel with low-resolution ADCs. 
{Left:} 1-bit ADC with constellation size $M=16$. 
{Right:} 1- and 2-bit ADC with $\nr=100$ antennas, comparing Jeffreys,
polynomial, and standard PAM constellations.}}
\label{fig:constellation_design}
\end{figure*}

\subsection{Approximate Jeffreys constellation}

In many channels of interest, $w_{J,\lambda^*}(\theta)$ does not admit a closed-form cdf,
making it difficult to compute $F^{-1}(\cdot)$ directly. To address this, we can
approximate $F^{-1}$ via a parametric function 
\[
  g_\xi(u),
  \quad
  \xi \in \Xi,
\]
and choose $\xi$ so that $g_\xi(\cdot)$ is ``close'' to the Jeffreys mapping~$F^{-1}(\cdot)$. 
We can then define an \emph{approximate Jeffreys constellation}
\begin{equation}
  \mathcal{X}_{{J}, \xi^*}
  :=
  \bigl\{\,c_P\,g_{\xi^*}(u):\;u\in\mathcal{S}\bigr\}, \label{eq:approx_Jeffreys}
\end{equation}
where $\xi^*$ is the fitted parameter vector, and $c_P\le 1$ again enforces the average-power constraint.

\revision{

In the following, we propose an approximation strategy that can be implemented via convex optimization.
Recall from~\eqref{eq:C(w)} that the capacity loss incurred by using an input distribution $w$
instead of the optimal Jeffreys prior $w_{J,\lambda^*}$ is
\[
D(w\, \|\, w_{J,\lambda^*})
\;-\;
\lambda^* \, \mathbb{E}_{\theta\sim w}\!\left[c(\theta)-P\right],
\]
where $\lambda^*\ge 0$ satisfies the complementary slackness condition
$ \lambda^*\,\mathbb{E}_{\theta\sim w_{J,\lambda^*}}\!\left[c(\theta)-P\right] = 0. $  
For any feasible $w$, we have $\mathbb{E}_{\theta\sim w}[c(\theta)-P]\le 0$, and hence
\[
D(w \,\|\, w_{J,\lambda^*})
-
\lambda^* \mathbb{E}_{\theta\sim w}[c(\theta)-P]
\;\ge\;
D(w \,\|\, w_{J,\lambda^*}).
\]
Conversely, when $D(w \,\|\, w_{J,\lambda^*})$ is small, Pinsker's inequality guarantees that
$w$ is close to $w_{J,\lambda^*}$ in total variation, and the penalty term
$\lambda^* \mathbb{E}_{\theta\sim w}[c(\theta)-P]$
is then close to
$\lambda^* \mathbb{E}_{\theta\sim w_{J,\lambda^*}}[c(\theta)-P] = 0$.
Therefore, we propose to solve the approximation problem 
\[
\min_{w\in\mathcal{W}}\ D(w \,\|\, w_{J,\lambda^*}). 
\qquad\qquad (P_1)
\]
We take $\mathcal{W}:=\{f_\xi : \xi\in\Xi\}$, where $f_\xi$ is a polynomial pdf of degree $d$, i.e.,
$f_\xi(\theta) = \sum_{i=0}^d \xi_i \theta^i.$
The feasible set~$\Xi$ enforces normalization and nonnegativity:
\[
\Xi := \left\{
\xi \in \mathbb{R}^{d+1} :
f_\xi \ge 0 \text{ on } [-A,A],\;
\int_{-A}^{A} f_\xi(\theta)\, d\theta = 1
\right\}.
\]
Using
\[
\alpha_i = \frac{A^{i+1} - (-A)^{i+1}}{i+1},
\]
we eliminate $\xi_0$ via
\[
\xi_0 = \frac{1 - \sum_{i=1}^d \xi_i \alpha_i}{\alpha_0}.
\]
Thus $\xi_1,\dots,\xi_d$ are the free optimization variables.

The approximation problem becomes
\[
\min_{\xi\in\Xi}\     D\bigl(f_\xi \,\|\, w_{J,\lambda^*}\bigr). 
\qquad\qquad (P_2)
\]

To handle the constraint $f_\xi(\theta)\ge 0$, we use a logarithmic barrier, as done in interior point methods.
Define the penalty as the divergence between the uniform distribution and $f_\xi$, i.e.,
\[
D(\mathbf{1}\,\|\,f_\xi)
=
\int_{-A}^{A}
\ln\!\left(\frac{(2A)^{-1}}{f_\xi(\theta)}\right)\frac{d\theta}{2A},
\]
which equals $+\infty$ when $f_\xi(\theta) \le 0$ on a positive-measure set.  
We obtain the penalized optimization problem parameterized by $\gamma > 0$ 
\[
\min_{\xi\in\Xi}
\;
D(f_\xi \,\|\, w_{J,\lambda^*})
+ \gamma\, D(\mathbf{1}\,\|\,f_\xi). 
\qquad\qquad (P_3(\gamma))
\]
Solving $P_3(\gamma)$ for a sequence of $\gamma\to 0$ yields the solution to $P_2$.
Define the objective function $\mathcal{L}(\xi) := D(f_\xi\|w_{J,\lambda^*}) + \gamma D(\mathbf{1}\|f_\xi)$, and 
\begin{equation}
\psi_\xi(\theta)
:= \ln f_\xi(\theta) + 1
   + \lambda^* c(\theta)
   - \tfrac12 \ln J(\theta)
   - \frac{\gamma}{2A}\frac{1}{f_\xi(\theta)}.
\end{equation}
The gradient and Hessian have the following analytical expression:
\begin{align}
\nabla_\xi \mathcal{L}(\xi)
&= \int_{-A}^{A}
   \nabla_\xi f_\xi(\theta)\,
   \psi_\xi(\theta)\,
   d\theta, \\
\nabla_\xi^2 \mathcal{L}(\xi)
&= \int_{-A}^{A}
   \nabla_\xi f_\xi(\theta)
   \nabla_\xi f_\xi(\theta)^\top
   \left(
      \frac{1}{f_\xi(\theta)}
      + \frac{\gamma}{2A}
        \frac{1}{f_\xi(\theta)^2}
   \right)
   d\theta.
\end{align}
We can verify that the Hessian is positive definite, which implies that $P_3(\gamma)$ is strictly convex. 
All integrals \newrevision{involved} above can be evaluated numerically using a sufficiently fine discretization of $[-A,A]$ and evaluating $c(\theta)$ and $J(\theta)$ at the grid points.  
In practice, only a few Newton steps are needed for the convergence to the optimal solution $\xi^*$ of $P_3(\gamma)$.
Given the optimal coefficients $\xi^*$, the cdf has a closed form:
\[
F_{\xi^*}(\theta)
=
\sum_{i=0}^d
\xi^*_i\,\frac{\theta^{i+1}-(-A)^{i+1}}{i+1}.
\]
The inverse cdf $F_{\xi^*}^{-1}(u)$ can then be obtained to machine precision $\varepsilon$
by solving $F_{\xi^*}(\theta)=u$, e.g., using bisection. Each iteration evaluates $F_{\xi^*}$ and $p_{\xi^*}$, giving a total cost\footnote{With Newton's method, the cost may be reduced to $O\!\left(d\,\log\log(1/\varepsilon)\right)$. } $O\!\left(d\,\log(1/\varepsilon)\right).$
Finally, letting $g_{\xi^*} = F_{\xi^*}^{-1}$, we construct the approximate Jeffreys constellation
via~\eqref{eq:approx_Jeffreys}.  
We refer to this as the \emph{approximate polynomial Jeffreys constellation}.

In fact, this approach extends directly to any linear family of approximators of the form $f_{\xi}(\theta) = \sum_{i=0}^{d} \xi_i \phi_i(\theta)$ provided that each basis function $\phi_i$ admits an antiderivative that can be evaluated in closed form, i.e., one can compute $\int_{-A}^{\theta} \phi_i(t) dt$ for any $\theta\in\Theta$ and $i=0,\ldots,d$. 
Beyond polynomial families, this includes, for example, trigonometric bases, wavelet-type bases, and many other sets of functions with tractable integrals.
}

Figure~\ref{fig:constellation_design} (right) illustrates how the Jeffreys
constellation, the approximate polynomial Jeffreys constellation, and a standard PAM compare in a SIMO
channel \revision{with 1- and 2-bit ADC} and $\nr=100$ antennas. Each constellation is used with a
uniform distribution subject to $P=A^2/4$. The Jeffreys
constellation and its approximate version achieve similar high rates, and both are significantly higher than
the one achieved by the standard PAM constellation.

\begin{remark}
  It is worth noting that in the numerical example above, we can compute the
  mutual information exactly because both the input and output alphabets are
  discrete and of manageable size.\footnote{Here, the sufficient statistic of
  the output follows a multinomial distribution, thanks to the i.i.d.~property
  of the channel.} In such a scenario, one might also consider optimizing the
  input distribution over a dense PAM grid. However, that strategy still requires
  enumerating the entire output space, an approach that becomes infeasible for
  higher-dimensional or more complicated~(e.g., continuous output) channels. 
  By contrast, our method relies only on the Jeffreys prior and thus avoids the ``curse of
  dimensionality'' in the output, providing a more scalable and systematic
  solution for constellation design in large-scale MIMO systems.
\end{remark}

\subsection{Multi-dimensional constellations}

When the input is not scalar, the constellation design can be more involved.
Nonetheless, one can often generalize the same techniques used in the
single-dimensional setting. 

For example, in a MIMO channel with isotropic fading, such as that in
Section~\ref{sec:app_MIMO}, the magnitude and direction of the transmitted
vector can be treated separately. From the Jeffreys prior, one obtains an
optimal radial (magnitude) distribution via a mapping \(F_{r}^{-1}\) of a
uniform grid \(\Sc_r\). Simultaneously, the directional component in
\(\mathbb{R}^{d}\) (or \(\mathbb{C}^{\nt}\) in real lifting) can be discretized
by a suitable tessellation of the unit sphere \(\Sc_v^{d-1}\). This approach is
conceptually similar to polar- or spherical-coordinate constellations often
employed in vector quantization and multi-antenna communications
\cite{ngo2019cube, love2003grassmannian}.  

More generally, one may generate the parameters
of interest $\theta_1,\ldots,\theta_d$ \emph{sequentially} using a chain rule on the probability
distributions:
\[
  F(\theta_i \mid \theta_1,\ldots,\theta_{i-1}),
\]
the conditional distribution of \(\theta_i\) given the previously selected
\(\theta_1,\dots,\theta_{i-1}\). By applying the inverse function
\(F^{-1}(u\mid \theta_1,\dots,\theta_{i-1})\) to a uniform grid in \(u\), one
obtains a discrete set for \(\theta_i\). Repeating this process for
\(i=1,\dots,d\) yields a multi-dimensional constellation aligned with the
Jeffreys prior in each conditional step. 

In either case, polynomial or similar approximations~(as discussed in the above scalar case) can be used to avoid explicit computation of
the inverse (conditional)~cdf. This strategy generalizes the scalar
design principles and may ensure that the resulting finite set of transmit vectors
continues to approximate the optimal Jeffreys prior well, even in higher
dimensions.

\revision{
\begin{remark}
  Without any additional structure that may reduce complexity, we must approximate for each $i=1,\ldots,d$, the conditional cdf $F({\theta_i} \mid \theta^{i-1})$ for each $\theta^{i-1} \in \mathcal{A}^{(i-1)}$ where $\mathcal{A}^{(i-1)}$ is the $(i-1)$-dimensional constellation previously constructed, whose size is $M^{i-1}$. This leads to a total complexity \newrevision{proportional} to $1+M+\cdots+M^{d-1} = \Theta(M^d)$. 
If, however, the model exhibits a Markovian structure of order $D$, then the approximation only needs to be performed for each $\theta_{i-D}^{i-1}$ in \newrevision{a} $D$-dimensional grid with size $M^D$, yielding a reduced complexity of $\Theta(M^D)$. 
In the fully general case, as the dimension $d$ grows, the complexity increases exponentially. Indeed, it is well known that sampling accurately \newrevision{from} high dimension distributions is computationally hard in general. 
\end{remark}
}

\section{Receiver Architecture}
\label{sec:receiver}

\subsection{Sufficient statistic and exact computation}

A \emph{sufficient statistic} of the output regarding the input is the (log-)likelihood function. 
Specifically, for a given parameter~(input) vector~$\thetav$ and output vector~$\yv=[y_1,\dots,y_\nr]$, we have
\begin{equation}
  \ln p(\yv \mid \thetav) = \sum_{k=1}^\nr \ln p(y_k \mid \thetav), \label{eq:llh}
\end{equation}%
where $\ln p(y_k \mid \thetav)$ is the per-antenna log-likelihood. Note that the log-likelihood~\eqref{eq:llh} only depends on the \emph{empirical distribution}~(also known as \emph{type}) of the output $\yv$, denoted by 
\begin{equation}
  \pi_{\yv}(y) := {1\over \nr} \sum_{k = 1}^{\nr} \ind(y_k = y), \quad y\in\Yc.  
\end{equation}%
Indeed, we can rewrite the log-likelihood function as
\begin{equation}
  \ln p(\yv \mid \thetav) = \nr \sum_{y\in\{\yv\}} \pi_{\yv}(y) \ln p(y_\ell\mid \thetav), 
\end{equation}%
where $\{\yv\}$ is the set of distinct entries in the vector $\yv$. 
Thus, the computational complexity of evaluating the log-likelihood is bounded by $\min\{\nr, |\Yc|\}$, depending on the output size and alphabet cardinality. 

In the finite output alphabet case with $|\Yc| \ll \nr$, the computational
complexity is limited.  For instance, in the extreme case with $1$-bit ADC at
each output, for the computation of the log-likelihood, we need to compute
essentially $p(0\mid\thetav)$ and $p(1\mid\thetav)$ with computational complexity only
depending on the input dimension $\nt$, but not on the output dimension $\nr$.

\subsection{Approximate log-likelihood via output quantization}

Often the output alphabet is large, multi-dimensional, or even unbounded, making the exact computation of the log-likelihood prohibitively complex with a cost proportional to $ \nr $. A practical solution is to reduce this complexity by quantizing the outputs into a finite number of bins.
Assume that the output space $ \Yc $ is partitioned into $ L $ disjoint bins $ \Vc_1, \ldots, \Vc_L $. Defining the empirical probability mass function as
\begin{equation}
  \pi_{\hat{\yv}}(\ell) := {1\over \nr} \sum_{k = 1}^{\nr} \ind(y_k \in \Vc_\ell), \quad \ell\in[L],
\end{equation}%
the log-likelihood can be approximated as 
\begin{equation}
  \ln p(\yv \mid \thetav) \approx \nr \sum_{\ell\in[L]} \pi_{\hat{\yv}}(\ell) \ln p(\ell \mid \thetav), \label{eq:llh_approx}
\end{equation}%
where we define $p( \ell\mid \thetav) := \int_{\Vc_\ell} p(y\mid \thetav) d y$, 
i.e., the probability that the output falls into bin~$\ell$. 
The approximation becomes better when the bin becomes smaller, i.e., when $L$ becomes larger. The complexity in \eqref{eq:llh_approx} is now proportional to $L$, independent of $\nr$. This is highly beneficial when $\nr$ is large or $\Yc$ is unbounded. 
\comment{Moreover, the \emph{storage complexity} --- i.e., the number of bits needed to store the type $\pi_{\hat{\yv}}$ --- is less than $L \log \nr$ bits, an improvement over storing the entire $\nr$-dimensional output that requires $\nr\log L$ bits. This reduction is relevant especially when the received signal must be transported for further processing rather than being decoded locally. }

\subsection{Achievable rate for continuous output alphabet and approximate log-likelihood}

In the following, we consider the case with continuous output alphabets, and investigate the capacity loss due to the approximation. 
First, we have the following result on the achievable rate when the receiver applies an approximate computation of the log-likelihood function. 
\begin{proposition}
  \label{prop:1}
  Assume the receiver computes the log-likelihood function with the approximation~\eqref{eq:llh_approx}, \comment{without the transmitter being aware of this approximation.}  Let $\Jm(\thetav)$ be the Fisher information of the original (continuous) output, and let $\Jm_{\!L}(\thetav)$ be the Fisher information induced by the quantized output using bins $\{\Vc_\ell\}_{\ell=1}^L$.  Define
\begin{equation}
  e_L \;:=\; \int_{\Theta} \ln\frac{\det(\Jm(\thetav))}{\det(\Jm_{\!L}(\thetav))} \, d\thetav.
  \label{eq:eL_def}
\end{equation}
Then, there exists a constant $ \kappa_0 $ (which depends on the parameter space $ \Theta $ and the power constraint $ P $) such that the achievable rate computed via the approximation in $ \eqref{eq:llh_approx} $ is within $\kappa_0 \, e_L$ bits of the asymptotic capacity $C(P)$.
\end{proposition}
\begin{proof}
  See Appendix~\ref{app:prop1}. 
\end{proof}
\comment{It is important to emphasize that the transmitter employs the same input distribution as if the receiver computed the log-likelihood exactly. Despite the mismatch, the achievable rate remains within a gap of $\kappa_0 e_L$ from the capacity $C(P)$. We see that $e_L$ is equivalent to the capacity loss up to a multiplicative constant $\kappa_0$. In the remainder of the section, we focus on deriving an upper bound on $e_L$. }

For clarity, let us rewrite explicitly $\Jm(Y; \thetav) = \Jm(\thetav)$, $\Jm(\hat{Y}; \thetav) = \Jm_{\!L}(\thetav)$. 
Since $\hat{Y}=\Q(Y)$ is a function of $Y$ and $y\mapsto \Q(y)$ does not depend on $\thetav$, the chain rule for Fisher information~(see, e.g., \cite{zamir1998proof}) implies that 
\begin{align}
  \Jm(Y;\thetav) &= \Jm(Y, \hat{Y}; \thetav) \\
  &= \Jm(\hat{Y}; \thetav) + \Jm({Y}; \thetav \mid \hat{Y}), 
\end{align}%
where $\Jm(Y;\thetav \mid \hat{Y}) := \sum_{\hat{y}} \mathbb{P}(\hat{Y} = \hat{y})\, \Jm(Y;\thetav \mid \hat{Y}=\hat{y})$ with $\Jm(Y;\thetav \mid \hat{Y}=\hat{y})$ defined as 
\begin{equation}
  \mathbb{E}\Bigl[\nabla_\thetav \ln p_\thetav(Y\mid \hat{Y}=\hat{y})\,\nabla_\thetav \ln p_\thetav(Y\mid \hat{Y}=\hat{y})^\T \,\Big|\,\hat{Y} = \hat{y} \Bigr]. 
\end{equation}%
Consequently,
\[
\ln\frac{\det(\Jm(Y;\thetav))}{\det(\Jm(\hat{Y};\thetav))}
=\ln\det\Bigl(I + \Jm(\hat{Y};\thetav)^{-1} \Jm(Y;\thetav \mid \hat{Y})\Bigr).
\]
Since for any positive semi-definite matrix $\bm{Q}$ the inequality $\ln\det(\Id+\bm{Q}) \le \operatorname{tr}(\bm{Q})$ holds,
we obtain the bound
\begin{equation}
  \ln\frac{\det(\Jm(Y;\thetav))}{\det(\Jm(\hat{Y};\thetav))} \le \operatorname{tr}\Bigl( \Jm(\hat{Y};\thetav)^{-1} \Jm(Y;\thetav \mid \hat{Y}) \Bigr). \label{eq:tmp722}
\end{equation}%
This upper bound vanishes with $\Jm(Y;\thetav \mid \hat{Y})$ in the following way. 
\begin{lemma}
  \label{lemma:phi_L}
  Suppose there exists a sequence $\{\phi_L\}_L$ with $\lim_{L\to\infty} \phi_L = 0$ so that $\Jm(Y;\thetav \mid \hat{Y}) \preceq \phi_L \Id$ for every $\thetav \in \Theta$. Then, for all $\thetav \in \Theta$ and for sufficiently large $L$ (specifically, when $\phi_L \le \lambda_{\min}/2$), it holds that 
  \begin{equation}
\ln\frac{\det(\Jm(Y;\thetav))}{\det(\Jm(\hat{Y};\thetav))} \le {2d \over \lambda_{\min}} \phi_L, \quad \forall\,\thetav\in \Theta,
  \end{equation}%
  where $\lambda_{\min}$ is the infimum of the minimum eigenvalue of $\Jm(Y;\thetav)$ over $\Theta$. 
\end{lemma}
\begin{proof}
Because $\Jm(Y;\thetav \mid \hat{Y}) \preceq \phi_L \Id$, by the decomposition of the Fisher information we have
$\Jm(\hat{Y};\thetav) = \Jm({Y};\thetav) - \Jm(Y;\thetav \mid \hat{Y})  \succeq \Jm({Y};\thetav) - \phi_L\,\Id \succeq (\lambda_{\min} - \phi_L)\,\Id$. From the positive semi-definiteness, we have $\Jm(\hat{Y};\thetav)\succeq (\lambda_{\min} - \phi_L)^{+}\,\Id$. Therefore, 
\[
\operatorname{tr}\Bigl( \Jm(\hat{Y};\thetav)^{-1} \Jm(Y;\thetav \mid \hat{Y}) \Bigr)
\le {d  \over (\lambda_{\min} - \phi_L)^+}\, \phi_L.
\]

When $\phi_L \le \lambda_{\min}/2$, we have $\lambda_{\min} - \phi_L \ge \lambda_{\min}/2$, implying
\begin{equation}
\operatorname{tr}\Bigl( \Jm(\hat{Y};\thetav)^{-1} \Jm(Y;\thetav \mid \hat{Y}) \Bigr) \le \frac{2d}{\lambda_{\min}}\, \phi_L.
\end{equation}%
Finally, substituting this bound into \eqref{eq:tmp722} completes the proof.  
\end{proof}
\comment{From Lemma~\ref{lemma:phi_L} and \eqref{eq:eL_def} in Proposition~\ref{prop:1}, if one can identify universal upper bound $\phi_L$ on the conditional Fisher information matrix $ \Jm(Y;\thetav \mid \hat{Y})$, then we have
\begin{equation}
  e_L \le \frac{2d}{\lambda_{\min}} \Vol(\Theta)\, \phi_L. \label{eq:tmp234} 
\end{equation}%
}

We now summarize the scaling of capacity loss due to approximate log-likelihood computation in the following theorem.
\comment{
\begin{theorem}
\label{thm:cond_fish}
Let $\{p_\thetav:\, \thetav\in\Theta\}$ be a family of densities on $\mathcal{Y}\subset\mathbb{R}^m$ that is twice continuously differentiable in $(\thetav,y)$ for $\thetav\in\Theta\subset\mathbb{R}^d$. Suppose there exists finite constants $\rho$ and $K_\epsilon$ such that
\begin{equation}
  \sup_{\thetav\in\Theta,\; y\in\mathcal{Y}} \Bigl\|\nabla_y\nabla_\thetav \ln p_\thetav(y)\Bigr\| \le \rho, \label{eq:lipschitz}
\end{equation}%
and 
\begin{equation}
\sup_{\thetav\in\Theta} \left(\E \left[ \|\nabla_\thetav \ln p_\thetav(Y) \|^{2+\epsilon} \right]\right)^{2\over 2+\epsilon}  \le K_\epsilon, 
\end{equation}%
where $\epsilon>0$ is such that \eqref{eq:tmp328} holds. 
For any $r>0$ and $L\in\mathbb{N}$, there exists a $(L+1)$-bin quantizer so that 
\begin{align}
  \Jm(Y;\thetav \mid \hat{Y}) \preceq  \left(  K_\epsilon P_{\thetav}(\|Y\| > r)^{\epsilon\over 2+\epsilon} + 36 \rho^2 r^2 L^{-2/m}  \right) \Id. 
\end{align}
Consequently, by optimizing over $r$, we obtain 
\begin{itemize}
  \item for bounded support, i.e., $\exists$~$r_0>0$ so that $\sup_{\thetav\in\Theta} P_{\thetav}(\|Y\| > r_0) = 0$, 
    \begin{equation}
      e_L \le A_0 L^{-2/m};  
    \end{equation}%
  \item for polynomial tail decay, i.e., $\exists$~$\kappa_1>0$ and $\eta_1>0$ so that $\sup_{\thetav\in\Theta} P_{\thetav}(\|Y\| > r ) \le \kappa_1 r^{-\eta_1}$, $\forall\,r>0$, 
    \begin{equation}
		e_L \le A_1\,L^{-\frac{2\eta_1\epsilon}{m(\eta_1\epsilon+2(2+\epsilon))}};
    \end{equation}%
  \item for Gaussian tail decay, i.e., $\exists$~$\kappa_2>0$ and $\eta_2>0$ so that $\sup_{\thetav\in\Theta} P_{\thetav}(\|Y\| > r ) \le \kappa_2 e^{- \eta_2 r^2}$, $\forall\,r>0$, 
    \begin{equation}
	e_L \le \left( A_2 + B_2 \ln L \right) L^{-2/m}; 
    \end{equation}%
\end{itemize}
where $A_0, A_1, A_2, B_2$ are some bounded positive values that only depend on the system parameters but not on $L$.  
\end{theorem}
}
\begin{proof}
  See Appendix~\ref{app:th2}. 
\end{proof}

\comment{
\begin{remark}
  Consider the polynomial tail decay scenario. The  result above reveals an interesting trade-off in the large $L$ regime between the three key aspects: storage complexity~($L\log \nr$~bits), computational complexity~$L$, and capacity loss~($\sim L^{-{2\over m} {\eta_1\epsilon\over 4+ 2\epsilon+\eta_1\epsilon}}$ bits). With Gaussian tail decay or bounded support, the scaling of the capacity loss improves to essentially $L^{-{2 \over m}}$ bits. 
\end{remark}
}

\revision{

\section{Extension to the Non-I.I.D. Setting}
\label{sec:noniid}
The main limitation of our results is the assumption that the output must be i.i.d.~given the input. In many practical large-scale MIMO systems, however, spatial correlation is inherent since the antennas may be close to each other.

To extend our results beyond the i.i.d.~setting, it would require a complete reworking of Clarke--Barron's proof techniques. More specifically, it mainly involves extending Laplace's method to a general $\nr$-letter setting and identifying verifiable sufficient conditions under which the divergence between the conditional output distribution and the marginal~(mixture) output distribution converges uniformly, in a similar fashion to Lemma~\ref{lemma:uniform} in the i.i.d.~case.  
Note that such uniform convergence is essential for establishing the asymptotic capacity~(maximum mutual information). 

A closer examination shows that, under suitable regularity conditions\footnote{We do not elaborate on these conditions here, as their full statement would take us too far afield and is out of the scope of the current work. Broadly speaking, they ensure uniform concentration of key quantities associated with the $\nr$-letter output distribution, such as the first and second derivatives of the log-likelihood.}, the asymptotic capacity remains unchanged when the Fisher information is replaced by the \emph{Fisher information rate}, whenever it exists:
\begin{equation}
  \Jm^{(\infty)}(\thetav) := \lim_{\nr\to\infty} \Jm^{(\nr)}(\thetav), 
\end{equation}
with
\begin{equation}
  \Jm^{(\nr)}(\thetav) := {1\over \nr} \E_{p_{\thetav}}\left[ \nabla_{\thetav} \ln p_{\thetav}(y^{\nr}) \left(\nabla_{\thetav} \ln p_{\thetav}(y^{\nr})\right)^\T \right], 
\end{equation}
where $y^{\nr} := \yv$ and $p_\thetav$ is the density of the random process $\{y_i\}_{i\ge1}$ and $p_{\thetav}(y^{\nr})$ is the corresponding marginal density on the first $\nr$ samples. 

\subsection{Example 1: Correlated CSI known at the receiver}
First, to connect the result to the practically relevant large-scale MIMO setting, let us consider the case where the channel matrix $\Hm \in \mathbb{R}^{\nt \times \nr}$ is known at the receiver, but exhibits spatial correlation across the receive antennas. As shown in the i.i.d.~case, we can let the effective output be $(y^{\nr},\Hm)$, where $\Hm$ is independent of $\thetav$ and, conditional on $\Hm$, the entries $y_i$'s are independent. It follows that   
\begin{align}
  \Jm^{(\nr)}(\thetav) &= {1\over \nr} \sum_{i=1}^{\nr} \E_{p_{\hv_i}}[\Jm(\thetav\mid \hv_i)], \label{eq:tmp1} 
\end{align}%
where $\hv_i$ denotes the channel coefficients of the $i$-th receive antenna and  
\begin{equation}
  \Jm(\thetav\mid \hv_i) :=  \E_{p_\thetav(\cdot|\hv_i)}\left[ \nabla_{\thetav} \ln p_{\thetav}(y_i \mid \hv_i) \left(\nabla_{\thetav} \ln p_{\thetav}(y_i \mid \hv_i )\right)^\T \right]. 
\end{equation}
When the $\hv_i$'s are correlated but have the same marginal distribution, the Fisher information rate exists and from \eqref{eq:tmp1} equals
\begin{equation}
  \Jm^{(\infty)}(\thetav) = \E_{p_{\hv}}[\Jm(\thetav\mid \hv)], 
\end{equation}
which is \emph{exactly} the same as in the i.i.d.~case. Therefore, perhaps surprisingly, the correlation does not affect the asymptotic capacity in the large receive antenna array regime, under the aforementioned conditions.  

To build intuition for this ``negative result'', let us recall the classical single-input multiple-output fading channel without peak-power constraint. We know that with normalized independent noise the exact ergodic capacity is 
\begin{align}
  C(P) &= {1\over 2}\E [\log(1+ P\|\hv\|^2)] \\
&= {1\over 2}\log \nr + \E\left[ \log\left({1\over \nr} + {P\over \nr}\|\hv\|^2\right) \right]
\end{align}%
with $\hv := [h_1, \ldots, h_{\nr}]$ a correlated channel vector. 
Now, if ${1\over \nr}\|\hv\|^2$ concentrates around its expectation and its lower tail is not too heavy, then the expectation converges when $\nr$ is large and  
\begin{equation}
  C(P) = { 1\over 2}\log \nr + \log \left( P\,\mathbb{E}\left[ h^2 \right] \right) + o(1)
\end{equation}
that depends only on the marginal distribution of $h_i$. This example aligns with our previous discussion and shows that in the large antenna array regime, the spatial correlation in receiver-known CSI does not change the asymptotic capacity. 

Nevertheless, the situation can be markedly different if the CSI is unknown at the receiver, or ${1\over \nr} \|\hv\|^2$ has a heavy lower tail near $0$, or the outputs are correlated given the CSI.

\subsection{Example 2: SIMO channel with correlated noise}
 
Let us now consider the SIMO case with correlated noise.  
\[
y_i = \theta + z_i,\qquad i=1,\dots,\nr,
\]
where the noise process $\{z_i\}_{i\ge1}$ is zero-mean, stationary, Gaussian with autocovariance function $\gamma(k), k\in\mathbb{Z}$, and power spectral density $N(f)$, $f\in[-\pi, \pi]$. 
It follows that the normalized Fisher information of the $\nr$-letter output distribution $p_\theta(y^{\nr})$ is
\begin{equation}
J^{(\nr)}(\theta) = {1\over \nr} \mathbf{1}_{\nr}^\top \pmb{\Sigma}_{\nr}^{-1}\,\mathbf{1}_{\nr},  \quad \theta \in \Theta, \label{eq:szego}
\end{equation}%
where $\pmb{\Sigma}_{\nr} := \bigl[ \gamma(j-k) \bigr]_{j,k\in[\nr]}$ is the covariance matrix of the noise vector $z^{\nr}$. 
Under mild regularity assumptions\footnote{The power spectral density should be strictly positive and continuous~\cite[Chap.~11]{Szego}.}, the right-hand side of \eqref{eq:szego} converges to $1\over {2\pi N(0)}$. Therefore, the \emph{Fisher information rate} exists:
\[
J^{(\infty)}(\theta)
:= \lim_{\nr\to\infty} \frac{1}{\nr}\,J^{(\nr)}(\theta)
= \frac{1}{2\pi N(0)}.
\]
Thus, in this correlated setting the Fisher information rate depends only on the low-frequency behavior of the noise through $N(0)$. Since $ 2\pi N(0) = \sum_{k=-\infty}^{\infty}\gamma(k)$, we see how the correlation structure influences the Fisher information rate. Interestingly, although the Jeffreys factor --- hence the asymptotic capacity --- depends on the correlation, the optimal Jeffreys prior~(i.e., the capacity-achieving input distribution) in this case remains unchanged. This is because the Fisher information rate is still constant in $\theta$, even in the presence of output correlation. 

Across both examples, we see that in order to apply our asymptotic results in the more general non-i.i.d.~setting, a practical approach is to work with the $\nr$-letter surrogate $\Jm^{(\nr)}(\thetav)$, which captures the correlation across receive antennas, rather than relying solely on the single-letter Fisher information $\Jm(\thetav)$.

}

\section{Conclusion}
\label{sec:conclusion}

We have presented a systematic application of Bayesian asymptotics to
analyze and design general large-scale MIMO channels, encompassing nonlinear
models, output impairments, and hardware constraints. In essence, the Fisher
information of the per-antenna output distribution emerges as the key of
asymptotic capacity, yielding a Jeffreys factor that both specifies the
optimal input distribution (the Jeffreys prior) and serves as a tractable
surrogate for the channel capacity. While the exact capacity is notoriously
difficult to compute for many nonlinear or large-dimensional channels, the
Jeffreys factor relies on the single-antenna output distribution and thus avoids the curse of dimensionality on the output side.

Although this work lays the foundation for a broad class of large-scale MIMO
problems, there are several natural extensions:
\begin{itemize}
  \item The proposed method readily extends to
  additional channel types, provided their single-letter distributions admit a
  smooth parameterization. These may include different fading models and noise
  statistics beyond those studied here.
  \item It would be interesting to extend the results to the multi-user setting, and identify the asymptotic capacity region. For instance, for multiple-access channels~(uplink), we would expect to apply the same approach to analyze the weighted sum capacity. 
\item \revision{ Our main results rely on the assumption of i.i.d.~outputs across antennas. 
Although we have outlined how the analysis may extend to the non-i.i.d.\ setting, a rigorous generalization of Clarke and Barron's result under verifiable and practically relevant conditions remains an important direction for future work.
 }
\end{itemize}

\appendix
\subsection{Regularity conditions for the parameterization}
\label{app:CB_cond}

We present the regularity conditions for the result in \cite{CB94} to hold. 
\begin{enumerate}
  \item[(i)] $\Theta$ is compact and contained in the interior of some region $\Omega$ with $\dim(\Omega) = \dim(\Theta)$.  
  \item[(ii)] The parameterization of the family $\{p_{\thetav}(y):\
    \thetav\in\Omega\}$ is one-to-one in $\Omega$.  
  \item[(iii)] The density $p_{\thetav}(y)$ is twice continuously
    differentiable in $\thetav$ for almost every $y$. For every
    $\thetav\in\Theta$, there is a
    $\delta = \delta(\thetav)$ so that for each $j,k\in[d]$
    \begin{equation}
      \E \left[ \sup_{\thetav'\in \Ball(\thetav,\delta)} \left| {\partial^2
      \over \partial\theta'_j \partial \theta'_k} \ln p_{\thetav'}(y)
      \right|^2 \right] \label{eq:tmp329} 
    \end{equation}%
    is finite and continuous as a function of $\thetav$ and for each
    $j\in[d]$  
    \begin{equation}
      \E  \left[ \left| {\partial
      \over \partial\theta_j } \ln p_{\thetav}(y)
      \right|^{2+\epsilon} \right]   \label{eq:tmp328}
    \end{equation}%
    is finite and continuous as a function of $\thetav$ for every
    $0\le\epsilon\le\epsilon_0$ for some $\epsilon_0>0$.  
\end{enumerate}

\begin{lemma}[Uniform convergence of the redundancy] \label{lemma:uniform}
  Let $\{p^{\nr}_\thetav, \thetav\in\Theta\}$, be a family of distributions, 
  $w(\thetav)$ be a positive and continuous prior over $\thetav\in
  \Theta$. Under the regularity conditions specified above, we have  
  \begin{align}
    \lim_{\nr\to\infty} \sup_{\thetav\in \Theta} \Biggl|
    &D(p^{\nr}_\thetav \| M_\nr^w)  \nonumber \\
    &-
    \left({d\over2}\log{\nr\over 2\pi e} + \log{\sqrt{\det(\Jm(\thetav))}
    \over w(\thetav)} \right) \Biggr| = 0, 
  \end{align}
  where $M_\nr^{w}(y^\nr) := \int_\Theta w(\thetav)
  \prod_{l=1}^\nr p_\thetav(y_l) d\thetav$ is the mixture distribution.
  \end{lemma}
\begin{proof}
  See \cite{CB94}.
\end{proof}
From the uniform convergence, we have the following convergence of the expectation. 
\begin{corollary}
  \label{coro:uniform}
When $\thetav\sim w$ and $\yv\sim p^{\nr}_\thetav$, we have 
  \begin{equation}
    I(\thetav; \yv) = {d\over2}\log{\nr\over 2\pi e} + \E_{\thetav\sim w} \left[ \log{\sqrt{\det(\Jm(\thetav))} \over w(\thetav)} \right] + o(1),   
  \end{equation}%
  when $\nr\to \infty$. 
\end{corollary}

\subsection{Proof of Theorem~\ref{thm:main}} 
\label{app:th1}

\begin{lemma}[Mixture distributions minimize Bayesian redundancy]
  \label{lemma:mixture}
  The mixture distribution
  $m^{w}:= \E_{\thetav\sim w} \left[ p_\thetav \right]$ is the unique distribution that minimizes the Bayesian
  redundancy $\E_{\thetav\sim w} \left[ D( p_\thetav \| R) \right]$ over all distributions $R$, i.e.,
  \begin{equation}
    \inf_{R} \E_{\thetav\sim w } \left[ D(p_\thetav \| R) \right] = \E_{\thetav\sim w }
    \left[ D(p_\thetav \| m^w) \right] = I(\thetav; Y),
  \end{equation}
  where $\thetav\sim w$ and $Y\sim p_\thetav$ conditional on $\thetav$. 
\end{lemma}
\begin{proof}
  Note that $D(p_\thetav \| R) = D(p_\thetav \| m^w) + \E_{p_\thetav} \left[ \log
  {m^w(Y) \over R(Y)} \right] $. Taking expectation over $\thetav$, we have 
  \begin{align}
    \lefteqn{\E_{\thetav\sim w} \left[ D(p_\thetav \| R) \right] } \qquad \nonumber \\
    &= \E_{\thetav\sim w} \left[ D(p_\thetav \| m^w) \right] + \E_{\thetav\sim
    w}\E_{p_\thetav} \left[ \log {m^w(Y) \over R(Y)} \right] \\
    &= \E_{\thetav\sim w} \left[ D(p_\thetav \| m^w) \right] + D(m^w \| R) \\
    &\ge \E_{\thetav\sim w} \left[ D(p_\thetav \| m^w) \right], 
  \end{align}%
  with equality if and only if $R = m^w$. The minimum redundancy
  coincides with the mutual information $I(\thetav; Y)$.  
\end{proof}
Recall that the capacity under the average-power constraint
\begin{equation}
  C(P) = \max_{w:\,\thetav\sim w, \E[c(\thetav)]\le P} I(\thetav; \yv). 
\end{equation}%
\comment{Since the constraint is convex, one can verify\footnote{Slater's
condition is verified.} that strong duality
holds, i.e., the maximization is equivalent to 
\begin{equation}
  C(P) = \min_{\lambda\ge0} \max_{w:\,{\thetav}\sim w} \left\{I(\thetav; \yv) -
  \lambda\, \E[\tilde{c}(\thetav)] \right\}, \label{eq:duality}
\end{equation}%
where $\tilde{c}(\thetav):=c(\thetav) - P$. 
For any prior distribution $w$, we have the following upper bound
\begin{align}
  I(\thetav; \yv) - \lambda\, \E[\tilde{c}(\thetav)] &= \inf_{R} \E_{\thetav\sim
  w} \left[ D(p_{\thetav}^{\nr}\|
  R) - \lambda\, \tilde{c}(\thetav) \right]  \\
  &\le \inf_{R} \sup_{\thetav} \Bigl( D(p_{\thetav}^{\nr} \| R) -
  \lambda\, \tilde{c}(\thetav) \Bigr) \\
  &\le \sup_{\thetav} \Bigl( D(p_{\thetav}^{\nr} \| m_\nr^{J,\lambda}) -
  \lambda\, \tilde{c}(\thetav) \Bigr), \label{eq:tmp199}
\end{align}%
where the first equality is from Lemma~\ref{lemma:mixture}; in the last
inequality, we set $R = m_\nr^{J,\lambda}$, i.e., a mixture of
$p_{\thetav}^{\nr}$ according to the following tilted Jeffreys prior
\begin{equation}
  w_{J,\lambda}(\thetav) := {2^{-\lambda\,\tilde{c}(\thetav)}
  \sqrt{\det(\Jm(\thetav))} \over \int_{\Theta} 2^{-\lambda\,\tilde{c}(\thetav)}
  \sqrt{\det(\Jm(\thetav))} d \thetav }, \quad \thetav \in \Theta.
  \label{eq:wJ} 
\end{equation}%
Applying Lemma~\ref{lemma:uniform} with the above prior to the upper
bound \eqref{eq:tmp199}, we obtain 
\begin{align}
  I(\thetav; \yv) - \lambda\, \E[\tilde{c}(\thetav)] &\le {d\over2}\log{\nr\over 2\pi e} \nonumber \\ 
  &\phantom{=} \ + \log  \int_{\Theta}
2^{-\lambda\,\tilde{c}(\thetav)}
\sqrt{\det(\Jm(\thetav))} d \thetav  + o(1), \label{eq:tmp111}
\end{align}
where we see that the right hand side does not depend on $w$ and is
therefore a capacity upper bound for any
$\lambda\ge0$. Let us use $\lambda^*$ that is the minimum $\lambda$ such
that $\E_{\thetav\sim w_{J,\lambda}} [\tilde{c}(\thetav)] \le 0$.  
From \eqref{eq:duality} and \eqref{eq:tmp111}, we have the following capacity upper bound. 
\begin{align}
  C(P) \le \bar{C}(P) &:= {d\over2}\log{\nr\over 2\pi e}  \nonumber \\ &\phantom{:=}\ + \log  \int_{\Theta}
2^{-\lambda^*\,\tilde{c}(\thetav)} \sqrt{\det(\Jm(\thetav))} d \thetav  + o(1). 
\end{align}%

To obtain the lower bound, let us fix a prior $w$ that is continuous and positive over $\Theta$, we have 
\begin{align}
    I(\thetav; \yv) &= {d\over2}\log{\nr\over 2\pi e} + \E_{\thetav\sim w} \left[ \log{\sqrt{\det(\Jm(\thetav))} \over w(\thetav)} \right] + o(1) \\
    &= {d\over2}\log{\nr\over 2\pi e} - D(w\| w_{J,\lambda^*}) \nonumber \\ &\phantom{=}\ + \E_{\thetav\sim w} \left[ \log{\sqrt{\det(\Jm(\thetav))} \over w_{J,\lambda^*}(\thetav)} \right] + o(1) \\
    &= \bar{C}(P) - D(w\| w_{J,\lambda^*}) + \lambda^* \E_{\thetav\sim w} [\tilde{c}(\thetav)] + o(1),
\end{align}%
which proves \eqref{eq:C(w)}. 
Choosing $w=w_{J,\lambda^*}$, we have $D(w\|w_{J,\lambda^*}) = 0$ and $\lambda^*\E_w[\tilde c(\thetav)] = 0$ where the latter is from 
$\lambda^*\E_{w_{J,\lambda^*}}[\tilde c(\thetav)]=0$ by the definition of $\lambda^*$. 
We have 
\begin{align}
   C(P) &\ge I(\thetav; \yv) 
   = \bar{C}(P) + o(1).
\end{align}%

Since the upper and lower bounds coincide with $\bar{C}(P)$ up to a vanishing term with
$\nr$, we have the asymptotic capacity \eqref{eq:capa1}. And the optimal
input distribution is $p^* = w_{J,\lambda^*}$ which remains the same
when we replace $\tilde{c}(\thetav)$ by $c(\thetav)$ in \eqref{eq:wJ}. 
}

\qed

\subsection{Proof of Proposition~\ref{prop:1}}
\label{app:prop1}

  When the receiver computes the log-likelihood with the approximation~\eqref{eq:llh_approx}, the channel becomes effectively a channel with quantized output, for which the asymptotically optimal input is $w_{J,\lambda^*_L}$. The asymptotic capacity of this channel with quantized output is 
  \begin{equation}
    C_L(P) = {d\over2} \log {\nr\over 2\pi e} + \log(\JF_L(\lambda^*_L)),  \label{eq:tmp986}
  \end{equation}%
with a vanishing term when $\nr$ is large; $\JF_L(\lambda_L)$ is the Jeffreys factor of the channel with quantized output. Note that the transmitter is not aware of the approximation, and applies the input distribution $w_{J,\lambda^*}$, the asymptotically optimal input for the original channel. The asymptotic achievable rate becomes, according to \eqref{eq:C(w)}, 
  \begin{align}
    C_L(P)
- D(w_{J,\lambda^*}\|w_{J,\lambda_L^*})
+ \lambda_L^*\, \mathbb{E}_{w_{J,\lambda^*}}\!\big[c(\thetav)-P\big].
    \label{eq:tmp887}
  \end{align}%
Then, we compute 
  \begin{align}
    \lefteqn{D(w_{J,\lambda^*} \| w_{J,\lambda^*_L}) } \nonumber\\
    &= \E_{\thetav\sim w_{J,\lambda^*}} \left[\log\left(  
    { {2^{-\lambda^*(c(\thetav) - P)}\over \JF(\lambda^*)} \sqrt{\det(\Jm(\thetav))} 
   \over {2^{-\lambda_L^*(c(\thetav) - P)}\over \JF_L(\lambda_L^*)} \sqrt{\det(\Jm_L(\thetav))} 
    }\right) \right] \\
    &= \log (\JF_L(\lambda_L^*)) - \log (\JF(\lambda^*)) \nonumber \\
    &\phantom{=}\ + (\lambda^*_L - \lambda^*) \E_{\thetav\sim w_{J,\lambda^*}}\left[ c(\thetav) - P \right] \nonumber \\
    &\phantom{=}\ + \E_{\thetav\sim w_{J,\lambda^*}} \left[\log {  \sqrt{\det(\Jm(\thetav))} 
    \over \sqrt{\det(\Jm_L(\thetav))}} \right], \label{eq:tmp989} 
  \end{align}%
  where the last expectation can be explicitly written as 
  \begin{multline}
\int_{\Theta} 
2^{-\lambda^*(c(\thetav)-P)}
\sqrt{\det \Jm(\thetav)}
\log\!\frac{\sqrt{\det \Jm(\thetav)}}
{\sqrt{\det \Jm_L(\thetav)}}
\, d\thetav
\\
\le
\kappa_0
\int_{\Theta}
\ln\!\frac{\det \Jm(\thetav)}
{\det \Jm_L(\thetav)}
\, d\thetav,
\label{eq:tmp988}
\end{multline}
  \newrevision{
  with 
  \begin{equation}
    \kappa_0 := \sup_{\thetav\in\Theta} \left\{ 2^{-\lambda^* (c(\thetav) - P)} \sqrt{\det(\Jm(\thetav))} \log e \right\}
  \end{equation}%
  being a bounded constant since the function inside the supremum is continuous and the set $\Theta$ is compact.
  }
  Plugging \eqref{eq:tmp986}, \eqref{eq:tmp989}, and \eqref{eq:tmp988} back in \eqref{eq:tmp887}, we obtain an achievable rate $C(P) - \kappa_0 e_L$ and complete the proof of the theorem. 

  \qed

\subsection{Proof of Theorem~\ref{thm:cond_fish}}
\label{app:th2}

\comment{
Fix $r>0$ and $L\in\mathbb{N}$. First, define $\Vc_{0} := \{y\in\Yc:\; \|y\| > r\}$. 
Then, we apply the following lemma to partition the remaining space into $L$ bins with diameter upper bounded by $\Delta := 6 r L^{-1/m}$. 
\begin{lemma}
  An $m$-ball of radius $r$ $\Bc_m(0,r)$ can be partitioned into $L$ bins $\Vc_\ell$, $\ell \in [L]$, such that each bin has diameter at most $\Delta := 6rL^{-1/m}$.
\end{lemma}
\begin{proof}
  Consider a maximum $\tfrac{\Delta}{2}$-packing of $\Bc_m(0,r)$, $v_1$, \ldots, $v_{L'}$, where $L'$ is the $\tfrac{\Delta}{2}$-packing number. That is,
  \begin{align}
    \min_{i\ne j \in [L']} \|v_i - v_j\| &\ge \Delta/2, \\
    \max_{y\in\Bc_m(0,r)} \min_{i \in [L']} \|y - v_i\| &< \Delta/2. 
  \end{align}%
  It is without loss of generality to consider $2r \ge \Delta$ for otherwise $L'=1$. From the first condition, we have 
  \begin{equation}
    \Vol(\Bc_m(0, r+\Delta/4)) \ge L'\, \Vol(\Bc_m(0,\Delta/4))
  \end{equation}%
  hence
  \begin{equation}
    {L'}^{1\over m} \le 4r/\Delta + 1 \le 6r/\Delta. 
  \end{equation}%
  From the second condition, we see that $v_\ell$, $\ell\in[L']$ is also a $\tfrac{\Delta}{2}$-covering, so that the Voronoi regions of $v_\ell$ have each diameter upper bounded by $\Delta$. Since $L'\le L$, one can always split the regions into $L$ bins without increasing the diameter upper bound. 
\end{proof}
}

The quantization function will therefore be $\hat{y} := \sum_{\ell=0}^{L} \ell\cdot \ind( y \in \Vc_\ell)$. 
Note that the conditional density is given by
\[
p_\thetav(y\mid \hat{Y}= \ell) = \frac{p_\thetav(y)}{\int_{\Vc_\ell} p_\thetav(u)\,du}, \quad y\in \Vc_\ell.
\]
Taking logarithms and differentiating with respect to $\thetav$, 
\begin{align}
  \lefteqn{\nabla_\thetav \ln p_\thetav(y\mid \hat{Y}=\ell)} \qquad \\ 
  &= \nabla_\thetav \ln p_\thetav(y) - \frac{\int_{\Vc_\ell} \nabla_\thetav p_\thetav(u)\,du}{\int_{\Vc_\ell} p_\thetav(u)\,du} \\
&= \nabla_\thetav \ln p_\thetav(y) - \E_{y'} \left[ \nabla_\thetav \ln p_\thetav(y') \mid \hat{Y} = \ell \right]. 
\end{align}%
Thus, for any $y\in {\Vc}_\ell$, we have
\begin{align}
  \lefteqn{ \|\nabla_\thetav \ln p_\thetav(y\mid \hat{Y}=\ell) \|}\qquad  \\
   &= \left\|\E_{y'} \left[\nabla_\thetav \ln p_\thetav(y) -  \nabla_\thetav \ln p_\thetav(y') \mid \hat{Y} = \ell \right] \right\| \label{eq:tmp544}\\
  &\le \E_{y'} \left[\left\| \nabla_\thetav \ln p_\thetav(y) -  \nabla_\thetav \ln p_\thetav(y') \right\| \mid \hat{Y} = \ell \right] \label{eq:tmp545}
\end{align}%
using the triangle inequality $\|\E[\xv]\| \le \E[\|\xv\|]$. 
By the definition of the conditional Fisher information, we have
\begin{multline}
\Jm(Y;\thetav \mid \hat{Y}=\ell) \\= \mathbb{E}\Bigl[\nabla_\thetav \ln p_\thetav(Y\mid \hat{Y}=\ell)\,\nabla_\thetav \ln p_\thetav(Y\mid \hat{Y}=\ell)^\T \,\Big|\, \hat{Y}=\ell\Bigr].
  \label{eq:cond_fish}
\end{multline}

\subsubsection*{The case $\ell\in[L]$}
Each bin has maximum diameter $\Delta$, implying  
\begin{equation}
  \left\| \nabla_\thetav \ln p_\thetav(y) -  \nabla_\thetav \ln p_\thetav(y') \right\| \le \rho\, \Delta, \quad y, y'\in \Vc_\ell. 
\end{equation}%
Therefore, from \eqref{eq:tmp545}, we have
 $\| \nabla_\thetav \ln p_\thetav(Y\mid \hat{Y}=\ell) \| \le \rho\,\Delta$ for all $Y\in B_{\ell}$. It follows that
\[
\nabla_\thetav \ln p_\thetav(Y\mid \hat{Y}=\ell)\,\nabla_\thetav \ln p_\thetav(Y\mid \hat{Y}=\ell)^\T \preceq  \rho^2\,\Delta^2\,\Id.
\]
Taking the expectation with respect to the conditional law, we have 
\[
\Jm(Y;\thetav \mid \hat{Y}=\ell) \preceq \rho^2\,\Delta^2\,\Id, \quad \ell\in[L]. 
\]

\subsubsection*{The case $\ell=0$}
Now consider the quantization bin $\Vc_{0}$. From \eqref{eq:cond_fish} and \eqref{eq:tmp544}, we have   
\begin{align}
	\Jm(Y;\thetav \mid \hat{Y}=\ell) &= \mathsf{Cov}\left[\nabla_\thetav \ln p_\thetav(Y) \,\Big|\, \hat{Y}=\ell\right]  \label{eq:cond_fish_2} \\
					    &\preceq \E \left[\nabla_\thetav \ln p_\thetav(Y) \nabla_\thetav \ln p_\thetav(Y)^\T \,\Big|\, \hat{Y}=\ell\right] \\
					    &\preceq \E \left[ \|\nabla_\thetav \ln p_\thetav(Y) \|^2  \,\Big|\, \hat{Y}=\ell\right] \Id. 
\end{align}%
Moreover,  
\begin{align}
  \lefteqn{\E \left[ \|\nabla_\thetav \ln p_\thetav(Y) \|^2  \,\Big|\, \hat{Y}= \ell\right] \mathbb{P}(\hat{Y} = \ell)}\nonumber \\ 
  &= \E \left[ \|\nabla_\thetav \ln p_\thetav(Y) \|^2\,  \ind(\hat{Y}=\ell) \right] \\ 
													&\le \E \left[ \|\nabla_\thetav \ln p_\thetav(Y) \|^{2+\epsilon} \right]^{2\over 2+\epsilon} \mathbb{P}(\hat{Y} = \ell)^{\epsilon \over 2+\epsilon}, \\
                                                                                                        &\le  K_{\epsilon} P_{\thetav}(\|Y\| > r)^{\epsilon\over 2+\epsilon} 
\end{align}%
using H\"older's inequality $\E[|AB|] \le \E[|A|^p]^{1\over p}\E[|B|^q]^{1\over q}$ for any $p,q\ge0$ such that ${1\over p} + {1\over q} = 1$. Therefore, we have  
\begin{align}
	\Jm(Y;\thetav \mid \hat{Y}) &\preceq  \left( \rho^2\,\Delta^2 + \kappa_{\epsilon} P_{\thetav}(\|Y\| > r)^{\epsilon\over 2+\epsilon} \right) \Id \\
  &=  \left( 36 \rho^2\, r^2 L^{-2/m} + K_{\epsilon} P_{\thetav}(\|Y\| > r)^{\epsilon\over 2+\epsilon} \right) \Id
\end{align}%
as claimed.  

Next, we consider different cases of the decaying of the tail probability $P_{\thetav}(\|Y\| > r)$. 
If $\sup_{\thetav\in\Theta} P_{\thetav}(\|Y\| > r_0) = 0$, then setting $r=r_0$ we have $\Jm(Y;\thetav \mid \hat{Y}) \preceq 36 \rho^2\, r_0^2 L^{-2/m} \Id$. 
If $\sup_{\thetav\in\Theta} P_{\thetav}(\|Y\| > r ) \le \kappa_1 r^{-\eta_1}$ then 
\begin{align}
  \lefteqn{
  \min_{r \ge 0} \left( 36 \rho^2\, r^2 L^{-2/m}+ K_{\epsilon} P_{\thetav}(\|Y\| > r)^{\epsilon\over 2+\epsilon} \right)} \nonumber \\ 
  &\le  \min_{r \ge 0} \left( 36 \rho^2\, r^2 L^{-2/m}+ K_{\epsilon} \kappa_1^{\epsilon\over 2+\epsilon } r^{-{ \eta_1\epsilon\over 2+\epsilon }} \right)\\
														 &= 
36\rho^2\frac{\eta_1'}{\eta_1\epsilon}\,L^{-\frac{2\eta_1\epsilon}{m \eta_1'}}
\left(\frac{\eta_1\epsilon\,K_{\epsilon}\kappa_1^{\epsilon/(2+\epsilon)}}{72(2+\epsilon)\rho^2}\right)^{\frac{2(2+\epsilon)}{\eta_1'}},
\end{align}
where $\eta_1' := \eta_1\epsilon+2(2+\epsilon)$. 
If $\sup_{\thetav\in\Theta} P_{\thetav}(\|Y\| > r ) \le \kappa_2 e^{- \eta_2 r^2}$ then 
\begin{align*}
  \lefteqn{\min_{r \ge 0} \left( 36 \rho^2\, r^2 L^{-2/m}+ K_{\epsilon} P_{\thetav}(\|Y\| > r)^{\epsilon\over 2+\epsilon} \right) }\nonumber \\ &\le  \min_{r \ge 0} \left( 36 \rho^2\, r^2 L^{-2/m}+ K_{\epsilon} \kappa_2^{\epsilon\over 2+\epsilon } e^{-{\eta_2 r^2 \epsilon\over 2+\epsilon}} \right) \\
														 &= 36 \rho^2 L^{-2/m} \left(\frac{2 + \epsilon}{\eta_2 \epsilon}\right) \left( 1 + \ln\left( \frac{K_{\epsilon} \kappa_2^{\epsilon\over 2+\epsilon } \eta_2 \epsilon}{36 \rho^2 (2+\epsilon)} L^{2/m} \right) \right). 
\end{align*}

\comment{
From the above bounds, we can identify $\phi_{L+1}$ in the different considered cases. Plugging it into \eqref{eq:tmp234}, we prove the theorem since $\phi_{L+1}$ and $\phi_{L}$ have the same scaling. To see this, if $e_{L+1} \le (A_2 + B_2 \ln L) L^{-2/m}$ for every $L$, then $e_L \le (A_2 + B_2 \ln (L-1)) (L-1)^{-2/m} \le (A_2 + B_2 \ln L) (L/2)^{-2/m} = (2^{2/m} A_2 + 2^{2/m} B_2 \ln L) L$. This is essentially the same by replacing $A_2$ and $B_2$ by $2^{2/m} A_2 \le 4A_2$ and $2^{2/m} B_2 \le 4B_2$, respectively. 
}

\bibliographystyle{IEEEtran}
\bibliography{IEEEabrv,./biblio}

\begin{IEEEbiographynophoto}{Sheng Yang}
	(Member, IEEE) received the B.E. degree in electrical engineering from Jiaotong University, Shanghai, China, in 2001, and both the engineer degree and the M.Sc. degree in electrical engineering from Telecom Paris, France, in 2004, respectively. In 2007, he obtained his Ph.D. from Université de Pierre et Marie Curie (Paris VI). From October 2007 to November 2008, he was with Motorola Research Center in Gif-sur-Yvette, France, as a senior staff research engineer. Since December 2008, he has joined CentraleSupélec, Paris-Saclay University, where he is currently a full professor. He has also hold visiting professorships in the University of Hong Kong (2015, 2016) and the Hong Kong University of Science and Technology (2023, 2024). He received the 2015 IEEE ComSoc Young Researcher Award for the Europe, Middle East, and Africa Region (EMEA). He was an associate editor of the IEEE transactions on wireless communications from 2015 to 2020. He is currently an associate editor of the IEEE transactions on information theory.
\end{IEEEbiographynophoto}

\begin{IEEEbiographynophoto}{Richard Combes}
(Member, IEEE) is currently an assistant professor in CentraleSupélec.
He received the Engineering Degree from Telecom Paris, France, in 2008,
the Master's Degree in Mathematics from University of Paris VII in 2009,
and the Ph.D. degree in Mathematics from university of Paris VI in 2013. 
He was a visiting scientist at INRIA (2012) and a post-doc in KTH (2013). 
He received the best paper award at SIGMETRICS 2019 and CNSM 2011. 
His current research interests are machine learning, networks and probability.
\end{IEEEbiographynophoto}

\end{document}